\documentclass[runningheads]{llncs}
\usepackage{graphicx}
\usepackage[ruled,vlined,linesnumbered]{algorithm2e}
\usepackage{amsmath,amssymb,txfonts,caption}
\usepackage{hyperref}
\begin{document}
\title{Rendezvous on a Known Dynamic Point on a Finite Unoriented Grid}
%
%\titlerunning{Abbreviated paper title}
% If the paper title is too long for the running head, you can set
% an abbreviated paper title here
%
\author{Pritam Goswami\orcidID{0000-0002-0546-3894\thanks{The first three authors are full time research scholars in Jadavpur University.}} \and
Avisek Sharma\orcidID{0000-0001-8940-392X}\and
Satakshi Ghosh\orcidID{0000-0003-1747-4037}\and
Buddhadeb Sau\orcidID{0000-0001-7008-6135}}
\authorrunning{P. Goswami, A. Sharma, S. Ghosh, and B. Sau}
% First names are abbreviated in the running head.
% If there are more than two authors, 'et al.' is used.
%
\institute{Jadavpur University, 188, Raja S.C. Mallick Rd,
Kolkata 700032, India\\
\email{\{pritamgoswami.math.rs, aviseks.math.rs, satakshighosh.math.rs,  buddhadeb.sau\}@jadavpuruniversity.in}}
\maketitle              % typeset the header of the contribution
\begin{abstract}
In this paper, we have considered two fully synchronous $\mathcal{OBLOT}$ robots having no agreement on coordinates entering a finite unoriented grid through a door vertex at a corner, one by one. There is a resource that can move around the grid synchronously with the robots until it gets co-located along with at least one robot. Assuming the robots can see and identify the resource, we consider the problem where the robots must meet at the location of this dynamic resource within finite rounds. We name this problem "Rendezvous on a Known Dynamic Point". 

Here, we have provided an algorithm for the two robots to gather at the location of the dynamic resource. We have also provided a lower bound on time for this problem and showed that with certain assumption on the waiting time of the resource on a single vertex, the algorithm provided is time optimal. We have also shown that it is impossible to solve this problem if the scheduler considered is semi-synchronous.

\keywords{Rendezvous \and Finite Grid \and Dynamic Resource.}
\end{abstract}
\section{Introduction}
A swarm of robots is a collection of inexpensive and simple robots that can do a task collaboratively by executing one single distributed algorithm. In recent days swarm robot algorithm has become an exciting topic for research for several different reasons. Firstly, from the economic perspective, it is in general cheaper than using powerful robots. Moreover, a swarm of robots can be easily scaled based on the size of the environment they are deployed. Also, a swarm of robots is more robust against different faults (eg. crash faults and byzantine faults). There are many other positive sides to using a swarm of robots for executing a task.  Thus, this topic has become quite relevant in the field of research and application. The application of swarm robots is huge. For example, it can be used for patrolling, different military operations, rescue operations, cleaning large surfaces, disaster management, network maintenance and there are several others. 
\subsection{Background and Motivation}
There are several tasks a swarm of robots can do like, arbitrary pattern formation (\cite{BAKS19}), gathering (\cite{CP05}), network exploration (\cite{OT18}), dispersion (\cite{DBS21}) and many more. Here, we are interested in the problem of gathering. Gathering is a very classical problem where a collection of robots deployed in an environment meets at a single point within a finite time. 
 This problem has been solved under different environments and different settings (\cite{CP05,DSKN16,DSN14,GSGS22,KKN10,KMP08,P07}).  Rendezvous is a special case of gathering where the number of robots that need to gather is exactly two (\cite{DFPSY16,FSVY16,HDT18,ISKIDWY12,SIYM99}). 
 
 Since the deployed robots are simple it is hard for them to exchange important information being far apart. So the main motivation for gathering is to meet at a single point where the robots can exchange information for doing some task. Now let the information is stored at a single point or a set of points in the environment. And the robots need to be on those specific points to exchange information. In that case, the robots must gather at one of those specific points for exchanging information.

Now, let there be one single point of resource in the environment and the resource itself is a robot doing some other task (i.e., the algorithm for the resource is independent of the algorithm presented in this paper) in the same environment and thus, can move freely until it meets with another robot. So, the question is can two robots gather at the location of this moving resource? this is the question that has been the motivation behind this paper. 

Now, it is quite obvious that the environment should be a bounded region otherwise it would be impossible to reach the resource. Also for a bounded region in a plane, finite point robots can't meet at the location of the resource as there are infinitely many empty points where the resource can move to avoid the meeting. Thus it is natural to consider this problem for a bounded network. Now a finite grid is a widely used network in various fields and has many real life applications. For that reason, we have considered a finite grid as the environment in this work. Also, This problem can be framed as the problem where two cops are chasing and catching a robber on the run,  on the streets of some city. Many cities have their road network in the pattern of a grid (e.g., Manhattan). For this reason also, studying this problem on a finite grid is interesting. 

Note that if two robots with weak multiplicity detection can gather at the location of the resource, in some bounded networks, then any number of robots can gather. This is because after two robots meet with the resource, the resource becomes still and the other robots simply move to the location of the resource. That is why we have considered this problem with two robots only, rather than using any number of robots.
\subsection{Earlier Works}
In this paper, we are focusing on the problem of rendezvous on a known dynamic vertex. Rendezvous is a special case of gathering involving two robots. Gathering has been studied under different environments and different models throughout the span of research on swarm robot algorithms. In \cite{CP05}, authors have shown that gathering on a plane is possible for fully synchronous $\mathcal{OBLOT}$ robots but in \cite{P07} it has been proved that for semi-synchronous and asynchronous $\mathcal{OBLOT}$ robots it is impossible to gather without any axis agreement and multiplicity detection capabilities. So considering multiplicity detection only a solution has been provided in \cite{CFPS12} under the asynchronous scheduler. Gathering has been studied under different networks also (\cite{DSKN16,DSN14,KKN10,KMP08}). In \cite{KMP08} Klasing et al. first proposed the problem on a ring and proved that gathering on a ring is impossible without the robots having multiplicity detection capabilities. In \cite{DSKN16}, the authors examined the problem on the grid and trees and they found out that gathering is impossible even with global multiplicity detection if the configuration is periodic or symmetric and the line of symmetry is passing through any of the grid lines. Considering limitations in the view of robots many works have been done recently in \cite{FPSW05,GSGS22,DUVY20,PS21}. Among these, the work in \cite{PS21} and \cite{GSGS22} considered infinite rectangular and triangular grids respectively.

Now Rendezvous is a special case of gathering which has been studied extensively in \cite{DFPSY16,FSVY16,HDT18,ISKIDWY12,SIYM99}. In \cite{SIYM99}, Suzuki et al. have shown that two $\mathcal{OBLOT}$ robots can't gather in a semi-synchronous setting if the robots do not have any agreement on their local coordinate system even with multiplicity detection. So in \cite{DFPSY16,FSVY16,HDT18} authors have solved the problem considering robots with $O(1)$ memory or $O(1)$ bits of message communication under an asynchronous scheduler.
\subsection{Our Contribution} 
Till now all work in gathering considered the meeting point to be not known from earlier. contrary to that, in this work it is assumed that, two fully synchronous robots entering a finite unoriented grid through a door at a corner of the grid, know the meeting point (i.e., can see and identify the resource). But the problem is, the meeting point (i.e the location of the resource)  can also move to an adjacent vertex along with the robots in a particular round. 

Assuming the robots to be of $\mathcal{OBLOT}$ model, a deterministic, distributed algorithm has been provided that solves the rendezvous problem on a known dynamic meeting point within $O(T_f \times (m+n))$ rounds, where $T_f$ is the upper bound of the number of consecutive rounds the meeting point i.e., the resource can stay at a single vertex alone and $ m \times n$ is the dimension of the grid. We have also shown that for solving rendezvous on a known dynamic point on a finite grid of dimension $m \times n$ at least $\Omega(m+n)$  epochs is necessary. Hence, if we assume that the maximum number of consecutive rounds, the location of the resource can stay the same is $O(1)$ then, the algorithm provided in this paper is time optimal. We have also proved that solving rendezvous on a known dynamic point on a finite grid is impossible if the scheduler considered is semi-synchronous. This justifies why a fully synchronous scheduler has been considered in this work.
\subsection{Organization of the Paper}
In section~\ref{Sec:2}, we have defined the problem formally and discussed the models of the robot, resource, and scheduler in detail. We also have some definitions and notations in this section which will be needed for the contents in Section~\ref{sec:4}. In Section~\ref{sec:3}, we have discussed the lower bound of time required to solve this problem and also proved an impossibility result about solving this problem under semi synchronous scheduler. In Section~\ref{sec:4}, we have described each phase of the algorithm with the correctness results mentioned in different theorems and lemmas. Finally, in Section~\ref{sec:5}, we conclude the paper with some future possibilities and pathways for this research to continue.
\section{Problem Definition and Model}
\label{Sec:2}
\subsection{Problem Definition}
Let $G$ be a finite grid of dimension $m \times n$. Suppose there is a doorway in a corner of the grid through which two synchronous robots $r_1$ and $r_2$ can enter the grid. The robots can only identify the door if they are located on it. Consider a movable resource that is placed arbitrarily on a vertex of $G$. Both robots can see the resource. The resource will become fixed if at least one of $r_1$ or $r_2$ is on the same vertex with the resource. Now the problem is to design a distributed algorithm such that after finite execution of which both the robots gather at the vertex of the resource.

\subsection{Model}
Let $G=(V, E)$ be a graph embedded on an euclidean plane where $V=\{(i,j) \in \mathbb{R}^2: i,j \in \mathbb{Z} , 0\le i< n , 0 \le j < m \}$  and there is an edge $e \in E$ between two vertices, say $(i_1,j_1)$ and $(i_2,j_2)$, only if  either $i_1 = i_2$ and $|j_1-j_2|=1$ or, $j_1 = j_2$ and $|i_1-i_2| =1$. We call this graph a finite grid of dimension $m \times n$. Though the graph is defined here using coordinates, the robots has no perception of this coordinates which makes this grid unoriented. A corner vertex is a vertex of $G$ of degree two. A vertex is called a boundary if either the degree of that vertex is three or the vertex is a corner.  $G$ has  four corner vertex among which exactly one corner vertex has a door. This vertex having a door is called  the door vertex. Robots can enter the grid by entering through that door. There is a movable resource, initially placed arbitrarily at a vertex $g_0$ ($g_0$ is not the door) of $G$. 

\subsubsection{Robot Model:}
The robots are considered to be
\begin{itemize}
    \item[$\blacksquare$] \textbf{Autonomous:} There is no centralized control.
    \item[$\blacksquare$]\textbf{Anonymous:} The robots do not have any unique identifiers for distinction.
    \item[$\blacksquare$] \textbf{Homogeneous:} All robots run the same distributed algorithm.
    \item[$\blacksquare$] \textbf{identical:} The robots are physically indistinguishable.
\end{itemize}
Also, the robots are considered to be point $\mathcal{OBLOT}$ robots (i.e., robots with no persistent memory).  The robots can enter through the door one by one. A robot can distinguish if a vertex is on the boundary or a corner of the grid. Also, a robot can identify the door only if it is on the door vertex. Observe that, if the robots could distinguish the door vertex from any other vertex while located on some arbitrary vertex of the grid, then an orientation of the grid can be agreed upon by the robots. But since that is not the case here there is no such orientation of the grid on which the robots can agree. This makes this model quite interesting. A robot can distinguish the resource from other robots. Each robot has its local coordinate system but they do not agree on any global coordinate system. 

The robots operate in a \textit{LOOK-COMPUTE-MOVE} (LCM) cycle. In each of the cycles, a robot that was previously idle wakes and does the following phases,

\textbf{\textit{LOOK}:} In \textit{LOOK} phase a robot takes a snapshot of its surroundings and gets the location of other robots and the resource according to its local coordinate system.

\textbf{\textit{COMPUTE}:} In this phase a robot performs an algorithm with the locations of resource and other robots as input and as an output of that algorithm it gets the location of a neighboring vertex called the destination point.

\textbf{\textit{MOVE:}} In \textit{MOVE} phase a robot moves to the destination point through the edge of $G$ joining its current location and destination vertex. It is assumed that no two robots can cross each other through one edge without collision.

After completion of \textit{MOVE} phase, the robot becomes idle  until it is
activated again.

The activation of the robots is controlled by an entity called a scheduler. In the literature, there are mainly three types of schedulers. In the following, we discuss all the scheduler models and the scheduler we have chosen among them for solving this problem.
\subsubsection{Scheduler Model:} There are mainly three types of schedulers that have been considered throughout the literature of swarm robotics. The models are as follows:

\textbf{Fully Synchronous Scheduler (FSYNC)} \begin{itemize}
    \item[$\Diamondblack$] Time is divided into rounds of equal lengths
    \item[$\Diamondblack$] At the beginning of each round all robots are activated.
    \item[$\Diamondblack$] In a particular round all activated robots perform the \textit{LOOK}, \textit{COMPUTE} and \textit{MOVE} phases together.
\end{itemize} 

\textbf{Semi Synchronous Scheduler (SSYNC)} \begin{itemize}
    \item [$\Diamondblack$]Time is divided into rounds of equal lengths
    \item [$\Diamondblack$] At the beginning of each round a subset of robots are activated.
    \item [$\Diamondblack$]In a particular round all activated robots perform the \textit{LOOK}, \textit{COMPUTE} and \textit{MOVE} phases together.
\end{itemize}

\textbf{Asynchronous Scheduler (ASYNC)} \begin{itemize}
    \item [$\Diamondblack$]There is no sense of rounds.
    \item [$\Diamondblack$]A robot can either be idle or in any of the \textit{LOOK}, \textit{COMPUTE}, or \textit{MOVE} phases while some other robots are activated.
\end{itemize}

In this work, we have shown that it is impossible to solve the problem of rendezvous on a known dynamic point if the scheduler is semi-synchronous or asynchronous. Hence considering a fully synchronous scheduler we have provided an algorithm \textsc{Dynamic Rendezvous} that solves the problem within finite rounds.

\subsubsection{Resource Model:} The resource $res$ is a movable entity, initially which is placed arbitrarily on a vertex (except the door) of $G$. The resource moves synchronously along with the robots. let the position of $res$ at round $i$ is denoted as $g_i$ ($g_0$ is the initial location). for some round $i$, $g_i$ and $g_{i+1}$ are at most 1-hop away. The movement of the resource $res$ is controlled by an adversary. So $g_{i+1}$ can be any neighbor of $g_i$. We assume that resource will stay fixed if it meets with at least a robot among $r_1$ and $r_2$. Otherwise, it can not stay fixed on a vertex forever. Let $T_f$ be the upper bound of the number of rounds that $res$ can stay fixed alone on a vertex of $G$. Also, it is assumed that the resource can not cross a robot on an edge without collision. Now if a robot and the resource collides on an edge then, the colliding robot would carry the resource to its destination vertex and then terminates.
\subsection{Notation and Definitions}
For a robot $r$ we denote the resource as $res$ and the other robot as $r'$. Now we have the following definitions. 
\begin{definition}[Door boundary of a robot] 
    If a robot $r$ is located on a boundary of the grid on which the door vertex is also located then that boundary is called the door boundary of the robot $r$ and is denoted as $BD(r)$.
\end{definition}
\begin{definition}[Perpendicular Line of robot $r$]
For a robot $r$ on a boundary, the straight line perpendicular to $BD(r)$ passing through $r$ is called the perpendicular line of robot $r$. It is denoted as $PD(r)$.
\end{definition}
% \begin{definition}[InitGather Configuartion]
% A configuration $\mathcal{C}$ is called  a \textsc{InitGather Configuartion} if there is one robot $r$ in $\mathcal{C}$ such that for the view of $r$ the following two condition holds.
% \begin{enumerate}
%     \item $r.x = res.x$ or $r.y = res.y$.
%     \item $r'.x \in \{res.x-1, res.x, res.x+1\} $ or $r'.y \in \{res.y-1, res.y, res.y+1\}$.
% \end{enumerate}
\begin{definition}[Distance from resource along $BD(r)$]
    Distance of the resource $res$ along boundary $BD(r)$ is defined as the hop distance of robot $r$ from the vertex $v$ on $BD(r)$ such that the line joining $v$ and $res$ is perpendicular to $BD(r)$. We denote this distance as $dist(r)$ for a robot $r$ on $BD(r)$.
\end{definition}
\begin{definition}[InitGather Configuartion]
    A configuration $\mathcal{C}$ is called  a \textsc{InitGather Configuartion} if:
    \begin{enumerate}
    \item two robots $r$ and $r'$ are not on same line.
     \item there is a robot $r$ such that $r$ and  the resource $res$ are on a grid line (say $L$).
     \item the perpendicular distance of the other robot $r'$ to the line passing through $res$ and perpendicular to $L$ is at most one.
        
    \end{enumerate}
\end{definition}
In the following Fig.~\ref{Fig:def1} and Fig.~\ref{Fig:def2} we have mentioned the entities we have defined above.
\begin{figure}[h!]
\begin{minipage}[ht]{0.45\linewidth}
\centering
\includegraphics[width=6cm, height=3.5cm]{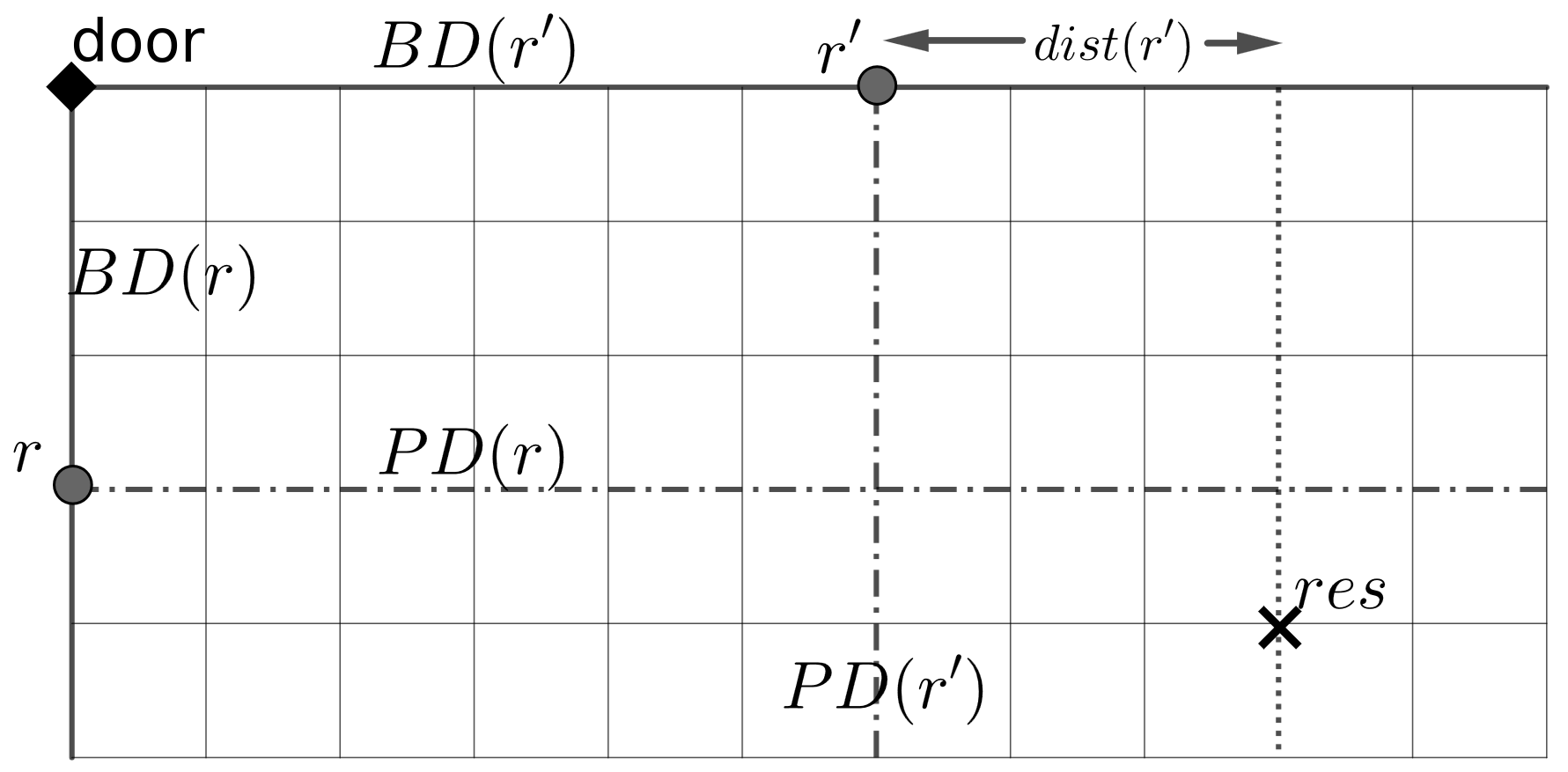}
     \caption{Diagram of a configuration mentioning $BD(r)$, $BD(r')$, $PD(r), PD(r')$ and $dist(r')$.}
     \label{Fig:def1}
\end{minipage}
\hfill
\begin{minipage}[ht]{0.45\linewidth}
\centering
\includegraphics[height=3.5cm, width=6cm]{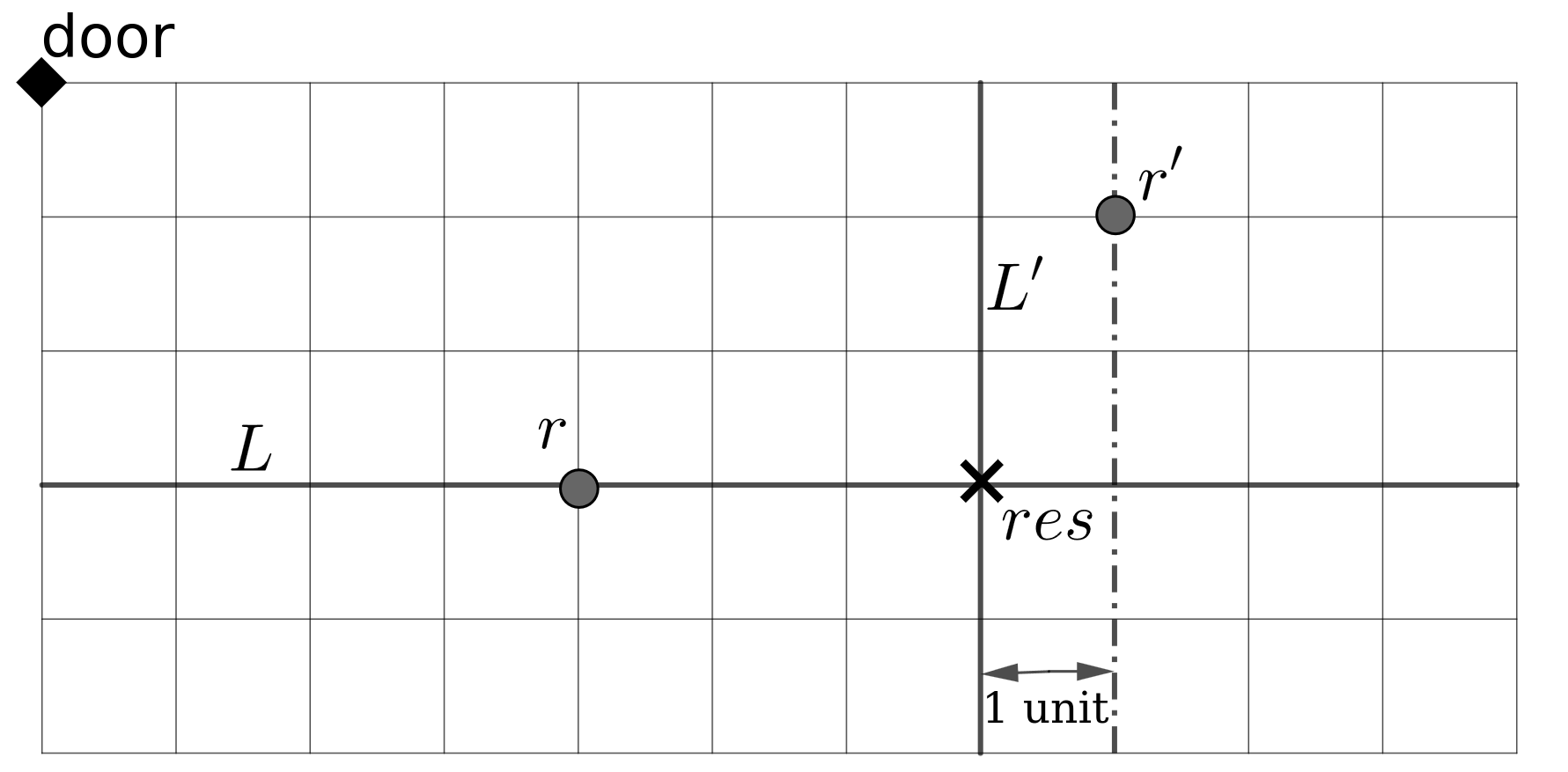}
    \caption{Diagram of an \textsc{InitGather Configuration}}
    \label{Fig:def2}
\end{minipage}
\end{figure}
\section{Lower Bound of Time and Impossibility}
\label{sec:3}
In this section, we will discuss the lower bound of time required to solve the problem of rendezvous on a known dynamic point on a finite grid of dimension $m \times n$. Also, we will prove an impossibility result which will justify our assumption of considering a fully synchronous scheduler to solve this problem. But first, let us define ``\textit{epoch}". An epoch is a time interval within which each robot in the system has been activated at least once. In the case of a fully synchronous scheduler, an epoch is equivalent to a round but for other schedulers, an epoch interval is finite but unpredictable. Now in the following theorem, we will discuss the time lower bound of solving rendezvous at a known dynamic point on a finite grid.
\begin{theorem}
    \label{theorem:timeLowerBound}
    Any algorithm that solves rendezvous at a known dynamic point on a finite grid of dimension $m \times n$ takes $\Omega(m+n)$ epochs in the worst case.
\end{theorem}
\begin{proof}
    Let us consider the scheduler to be a fully synchronous scheduler. Thus an epoch is equivalent to a round.
    Consider the following diagram (Fig.~\ref{Fig:lowerBoundTime}) where after each $T_f$ consecutive rounds, the resource changes its location from either $P$ to $Q$ or $Q$ to $P$. This implies after entering from the door vertex the robots must meet the resource either in vertex $P$ or in vertex $Q$. Now from the door vertex, the shortest path to $P$ or $Q$ is of length $m+n-1$. So to meet at either $P$ or $Q$ with the resource, each robot must travel through a path of length at least $m+n-1$. Now since in a round a robot can only move a path of length one, to travel a path of length $m+n-1$ at least $m+n-1$ round i.e epoch is necessary to solve this problem. Hence the result.
    \qed
\end{proof}

\begin{figure}[ht]
    \centering
     \includegraphics[height=4cm,width=7cm]{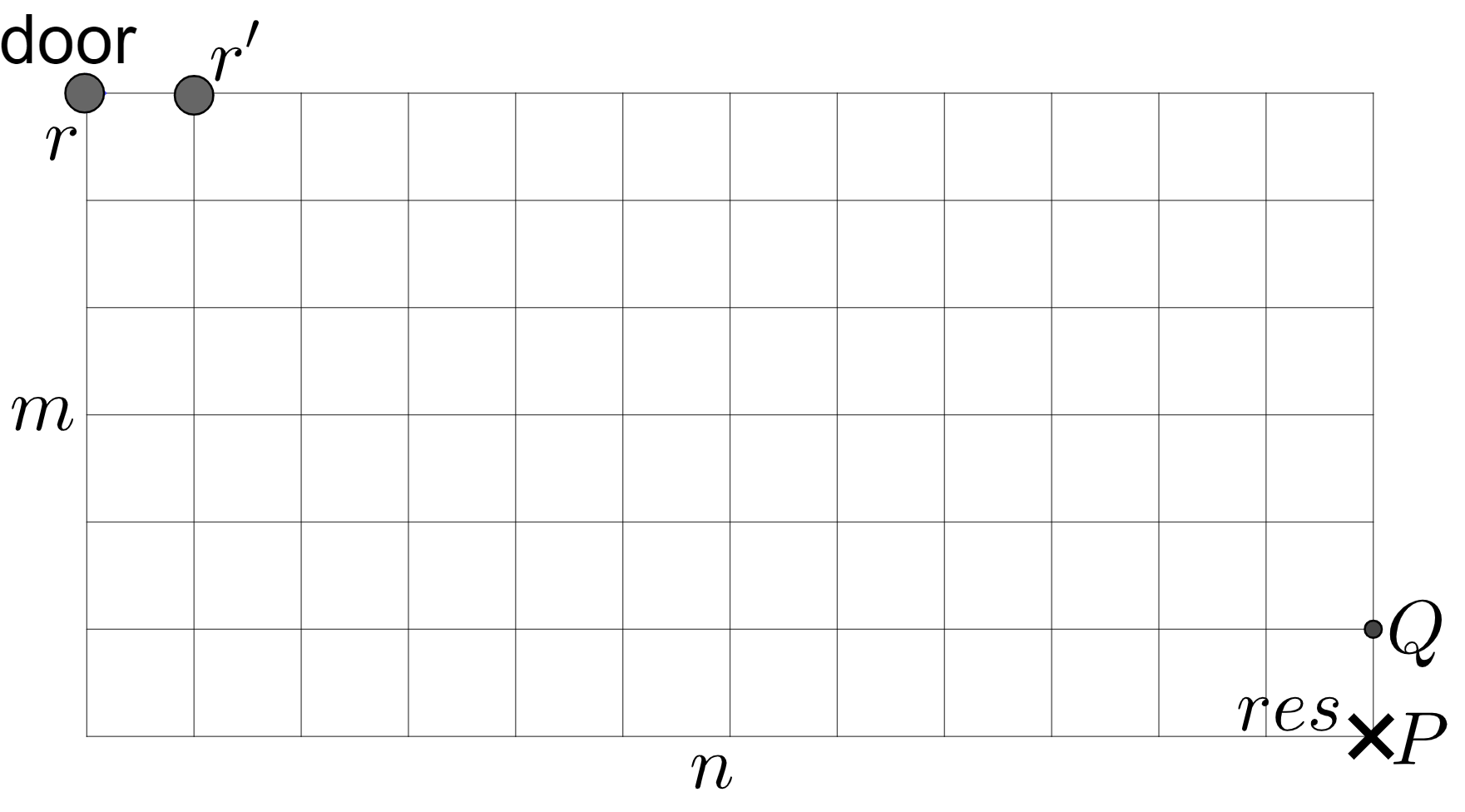}
     \caption{from the door vertex to reach $P$ or $Q$ the robot $r$ needs to travel at least a path of length $m+n-1$.}\label{Fig:lowerBoundTime}
    \end{figure}
Now we will discuss the impossibility result in the next theorem.
\begin{theorem}
    \label{theorem:impossible}
    No algorithm can solve the problem of rendezvous on a known dynamic point on a finite grid of dimension $m \times n$ if the scheduler is semi-synchronous.
\end{theorem}
\begin{proof}
    Let there is an algorithm $\mathcal{A}$ such that after finite execution of which two robots on a finite grid of dimension $m \times n$ meet at the location of the dynamic resource. Let $m, n >2$. Also, let $t$ be the round such that after completion of which at least one robot reaches the location of the resource and terminates.

    Let no robot is adjacent to the resource at the beginning of round $t$. This implies at the beginning of round $t$, the resource has at least two empty neighbor vertices. Now let the adversary activates only one robot during this round. Thus, even if the activated robot moves to one of the resource's empty adjacent vertex, another empty vertex remains empty. So even if the resource has to move during round $t$ it can always find an empty vertex to move that remains empty after the completion of the round. Hence after completion of round $t$, no robot can move to the location of the resource. Thus we reach a contradiction. Now, let exactly one robot is adjacent to the resource $res$ at the beginning of round $t$. Then, at least $res$ has one empty vertex which is not reachable by the adjacent robot in one round. So, if the adversary activates only the adjacent robot, say $r$, and $res$ moves to the empty vertex not reachable by $r$ then again we reach a contradiction.  Hence both the robots must be adjacent to the resource at the beginning of round $t$. Now if the resource is not at the corner and both the robots are adjacent to the resource at the beginning of round $t$ then, the resource must have an empty adjacent vertex that is not reachable by the robots in one round. Thus if the resource moves to that vertex during round $t$, we again reach a contradiction.

    \begin{figure}[h!]
\label{Fig:Impo}
\begin{minipage}[ht]{0.5\linewidth}
\centering
\includegraphics[width=5cm, height=3cm]{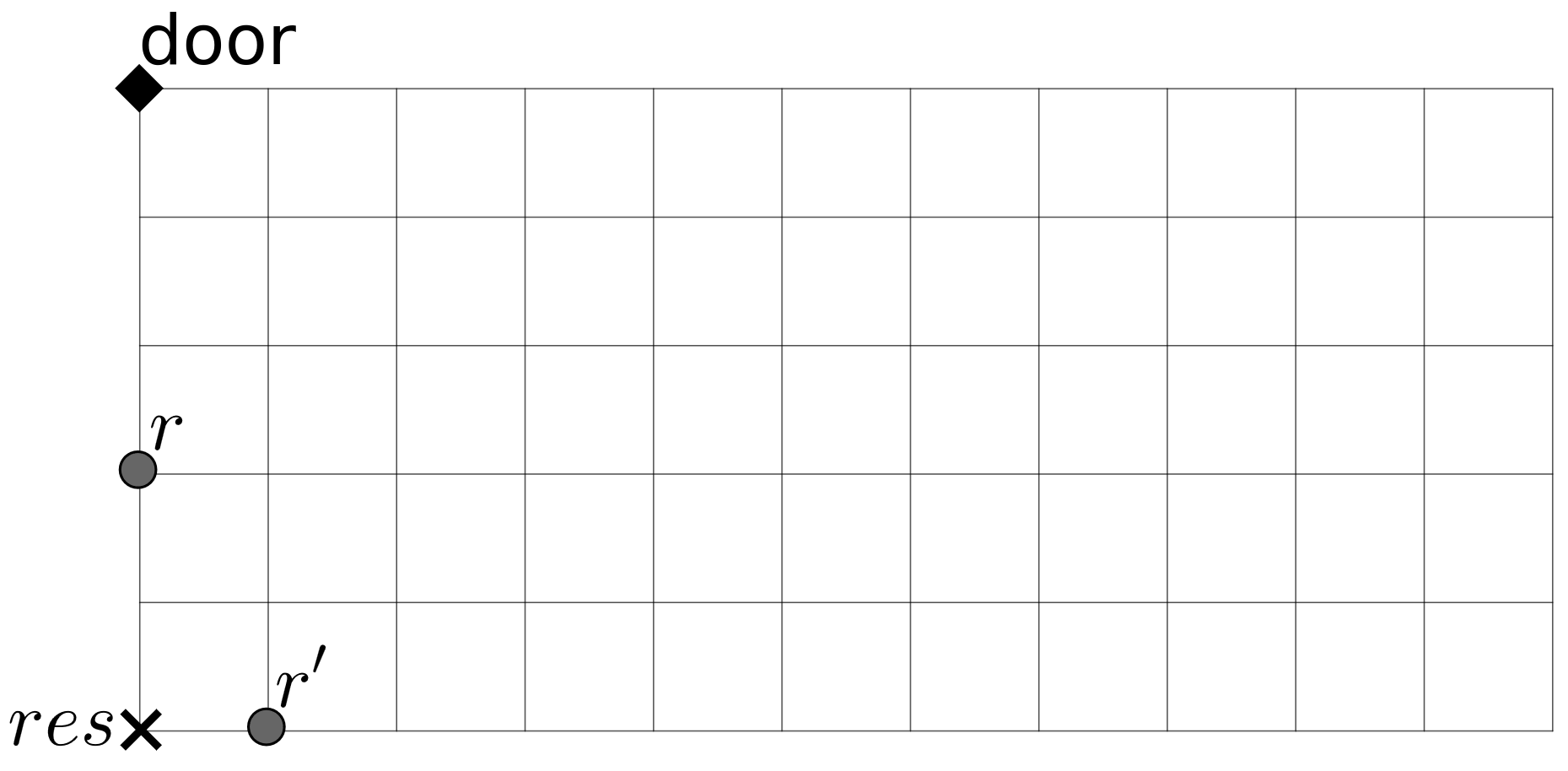}
     \captionsetup{labelformat=empty}
     \caption{ Configuration $\mathcal{C}_1$}
     \label{Fig:c1}
\end{minipage}
\hfill
\begin{minipage}[ht]{0.48\linewidth}
\centering
\includegraphics[height=3cm, width=5cm]{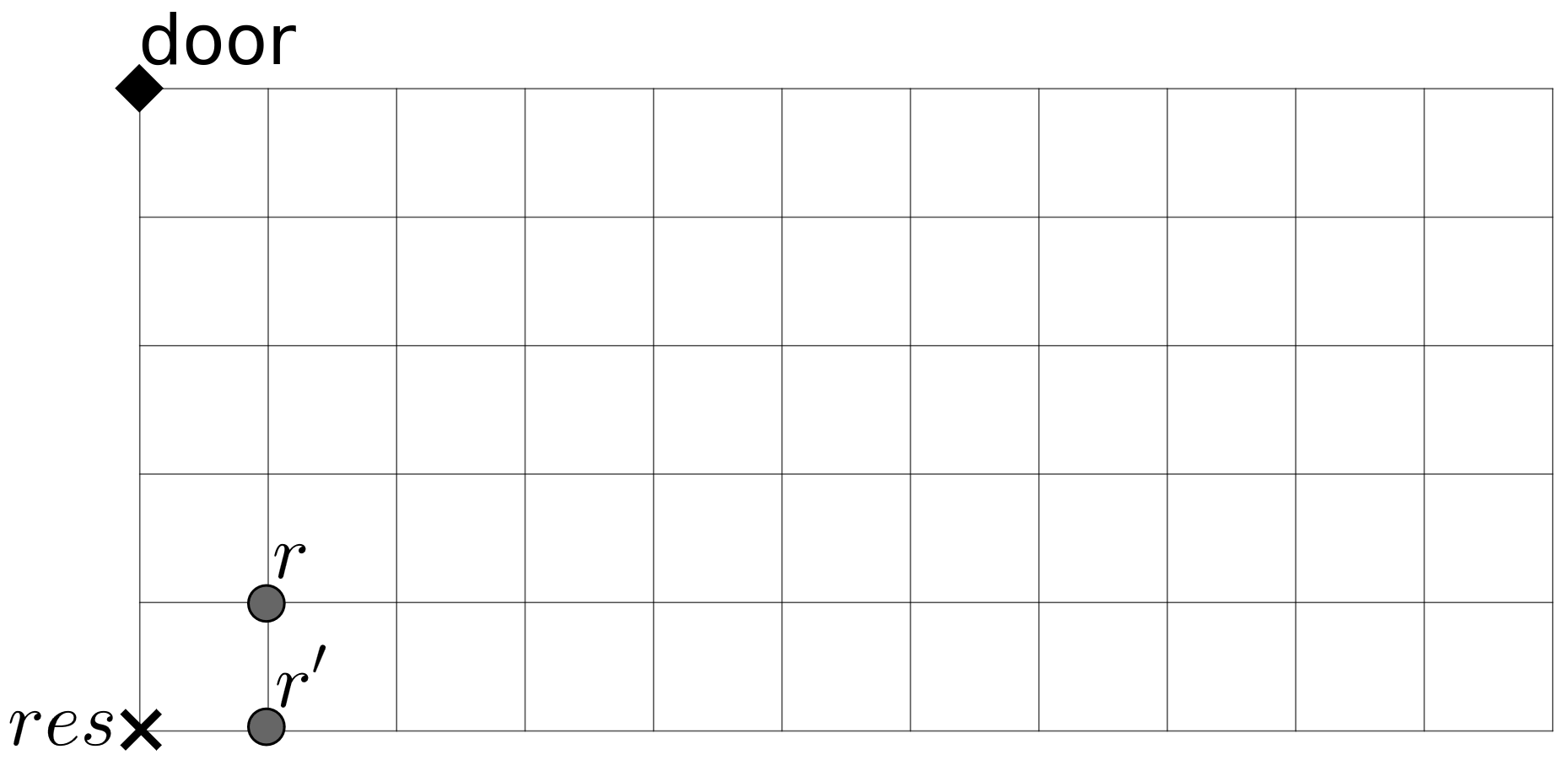}
    \captionsetup{labelformat=empty}
    \caption{Configuration $\mathcal{C}_2$}
    \label{Fig:c2}
\end{minipage}
\begin{minipage}[ht]{0.5\linewidth}
\centering
\includegraphics[width=5cm, height=3cm]{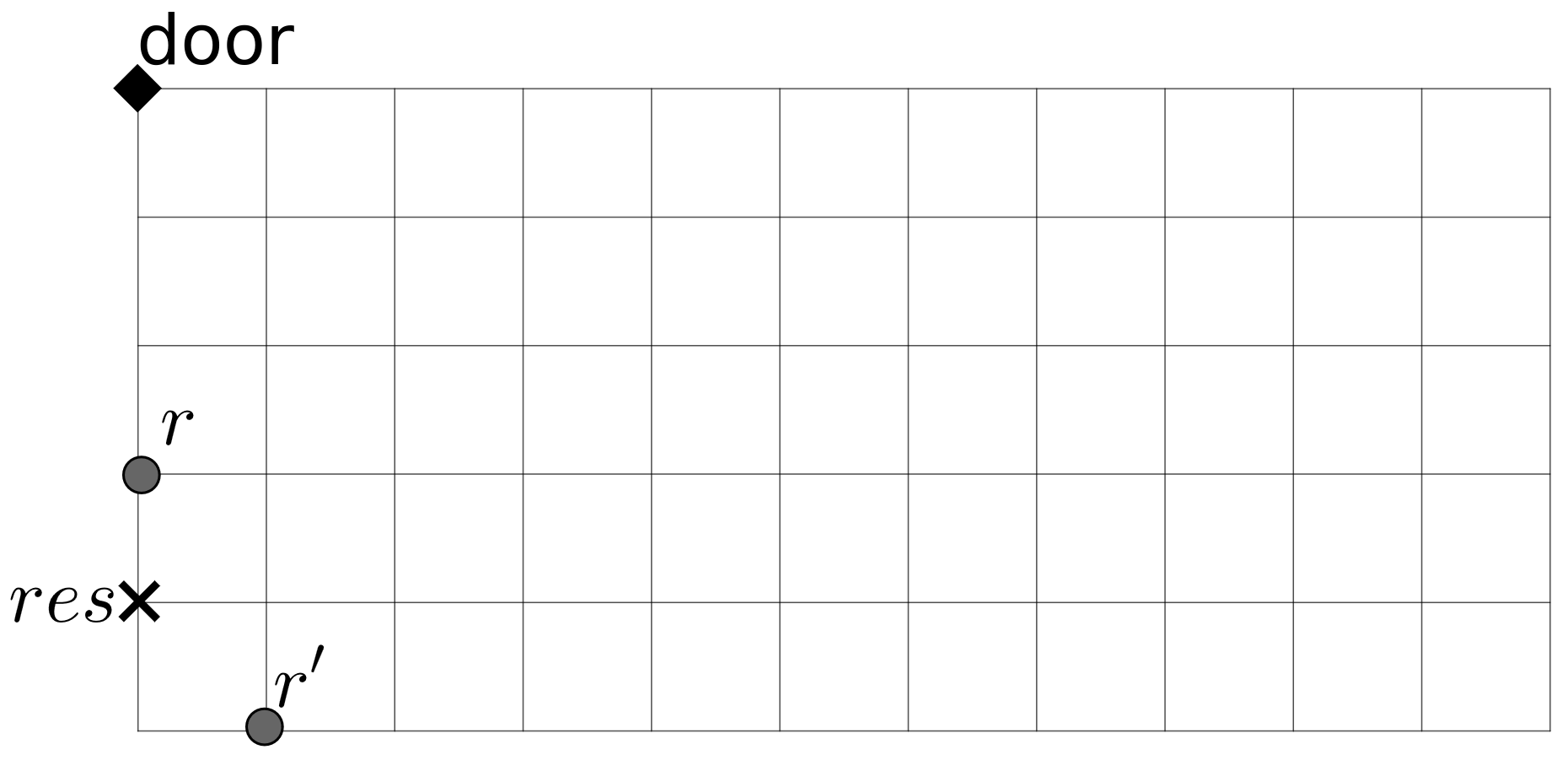}
     \captionsetup{labelformat=empty}
     \caption{ Configuration $\mathcal{C}_3$}
     \label{Fig:c3}
\end{minipage}
\hfill
\begin{minipage}[ht]{0.5\linewidth}
\centering
\includegraphics[height=3cm, width=5cm]{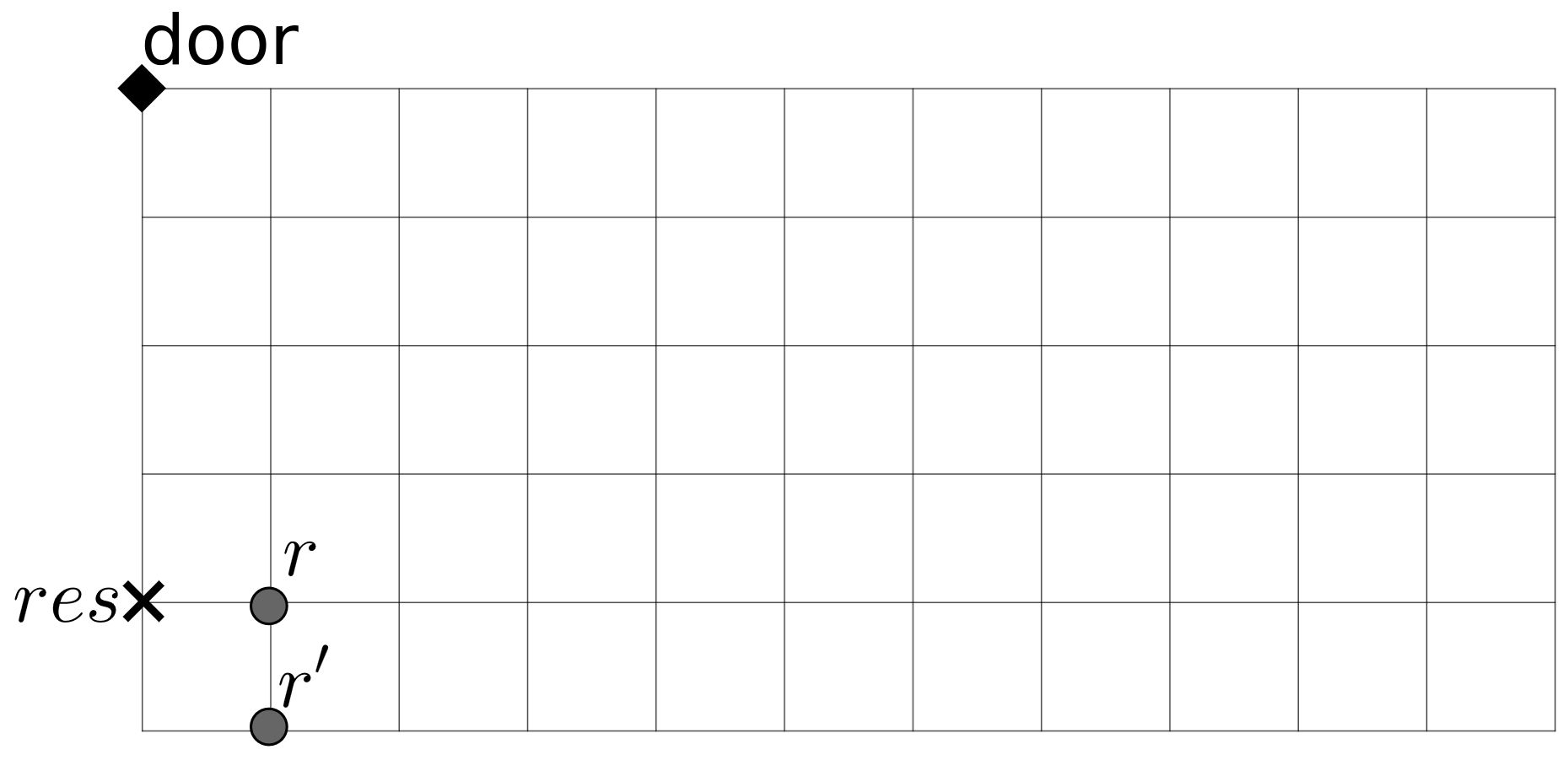}
    \captionsetup{labelformat=empty}
    \caption{Configuration $\mathcal{C}_4$}
    \label{Fig:c4}
\end{minipage}
\setcounter{figure}{3}
\caption{Some examples of configuration $\mathcal{C}_{corner-1}$. }
\end{figure}
    Now if we can prove that the configuration (say, $\mathcal{C}_{corner}$) where the resource is at a corner and both the robots are adjacent to it, is never formed then we are done. Let the adversary always activates only one robot in a particular round. Now, if possible let at the beginning of round $t$, the configuration is $\mathcal{C}_{corner}$. This implies the configuration, say, $\mathcal{C}_{corner-1}$, that was formed just before   $\mathcal{C}_{corner}$ must be one of $\mathcal{C}_1$, $\mathcal{C}_2$, $\mathcal{C}_3$ or $\mathcal{C}_4$ (Fig.4). Since adversary is compelled to activate only one robot in a particular round, so in configuration $\mathcal{C}_{corner-1}$ one the robot must be adjacent to the corner vertex. Without loss of generality let $r'$ be that robot. Note that, in $\mathcal{C}_1$ and $\mathcal{C}_2$ $res$ did not move to form $\mathcal{C}_{corner}$ and, in $\mathcal{C}_3$ and $\mathcal{C}_4$ $res$ had to move to form $\mathcal{C}_{corner}$. 

    Note that in all of these configurations, there is only one robot that is adjacent to the resource. If the adversary activates the adjacent robot then, in all of these configurations the resource can find an empty adjacent vertex that is not a corner and remains empty even after the move of the resource. Thus from any of the four configurations $\mathcal{C}_1$, $\mathcal{C}_2$, $\mathcal{C}_3$ and $\mathcal{C}_4$, $\mathcal{C}_{corner}$ is not formed. Hence we arrive at a contradiction. Thus $\mathcal{C}_{corner}$ will never be formed and hence the result.\qed 
\end{proof}

This justifies the necessity of a fully synchronous scheduler to solve this problem. In the next section assuming a fully synchronous scheduler, we have provided an algorithm that solves this problem of rendezvous on a known dynamic point on a finite grid.
\section{Algorithm}
\label{sec:4} 
It is quite obvious to observe that without the help of the other robot, a robot can not independently reach the location of the resource if the resource is controlled by an adversary. So to solve this problem the two robots must work together collaboratively and push the resource toward a corner. This is the main idea that is used to develop the proposed algorithm. Now since there is no agreement on the coordinates of the robots and the robots are oblivious, the main challenge here is to agree on the direction for the robots to move.

The rendezvous algorithm \textsc{Dynamic Rendezvous}, proposed in this section is executed in three phases. \textsc{Entry Phase}, \textsc{Boundary Phase} and \textsc{Gather Phase}. In the \textsc{Entry Phase}, the robots move in the grid one by one through the door vertex. This phase ends when the robots are located on the two adjacent vertices of the door vertex. Then in \textsc{Boundary Phase} the robots move along their corresponding boundary to form a special kind of configuration called \textsc{InitGather Configuration} (Definition.~\ref{def:def2}). Then in \textsc{Gather Phase} the robots move maintaining the \textsc{InitGather Configuration} and pushing the resource to a corner. 

In the first two phases, the agreement on the direction of movement for the robots is constructed from the fact that the robots know the location of the resource and identify vertices on the boundaries and corners of the grid. In these two phases, it is ensured that the robots are on the two boundaries of the grid of which the door is a part and always remain on their corresponding boundaries. 
In the \textsc{Gather Phase} though,  the robots move inside the grid leaving its boundary. In this situation as the robots are not on boundaries, they can not decide on a specific boundary for agreement. In this scenario, the agreement on the direction comes from the fact that at least one robot must be on a line along with the resource during each round of this phase.  After this brief overview of the algorithm let us describe it in detail. The algorithm \textsc{Dynamic Rendezvous} is as follows.

\begin{algorithm}[H]
\caption{Dynamic Rendezvous}\label{algo:main}
\textbf{Input:} A configuration $\mathcal{C}$.

\textbf{Output:} A destination point of robot $r$.

\eIf{a robot, say $r$, is at a corner $\land$ no robot terminates $\land$ (there is no other robot on the grid $\lor$ another robot is adjacent to $r$ )}
{
    Execute \textsc{Entry Phase}\;
}
{
    \eIf{$\mathcal{C}$ is \textsc{InitGather Configuration}}
    {
        Execute \textsc{Gather Phase}\;
    }
    {
        Execute \textsc{Boundary Phase}\;
    }
}
 
\end{algorithm}

The three phases are described in more detail in the following subsections.

\subsection{\textsc{Entry Phase}}
The first phase is called the \textsc{Entry Phase} (Algorithm~\ref{algo:EntryPhase}). During this phase, both the robots enter through the door vertex one by one into the grid $G$.

 \begin{algorithm}[H]
\caption{Entry Phase for robot $r$}\label{algo:EntryPhase}
\eIf{$r$ is on door vertex}
{
    \eIf{no other robot on boundary}
    {
        move through any edge on the boundary\;
    }
    {
        move through the edge on the boundary where there is no other robot\;
    }
}
{
    \If{ the other robot $r'$ is at a corner}
    {
        does not move\;
    }
}
 
\end{algorithm}
 A robot on the door vertex first checks if it can see another robot already on the grid. If it does not find any other robot on the grid, it moves to any of the two adjacent vertices of the door vertex. On the other hand, if there is already a robot on an adjacent vertex of the door vertex then, the robot on the door vertex moves to the other adjacent vertex of the door vertex. Note that, a robot on one adjacent vertex of the door vertex does not move and does not start executing any other phases if it sees another robot in the corner adjacent to it.

The \textsc{Entry Phase} ends when both robots are at the two distinct adjacent vertices of the door vertex in the corner. After the \textsc{Entry Phase} the robots will check if the configuration is an \textsc{InitGather Configuartion} or not. If the configuration is not an \textsc{InitGather Configuartion} then the robots execute the \textsc{Boundary Phase}, otherwise, they execute the \textsc{Gather Phase}.

\subsection{\textsc{Boundary Phase}}
The \textsc{Boundary Phase} starts after the end of the \textsc{Entry Phase} when the configuration is not an \textsc{InitGather Configuration} and no robots are at corners. In the initial configuration of \textsc{Boundary Phase}, both the robots are at the boundary and on the two distinct adjacent vertices of the door vertex. It is ensured in the algorithm of this phase (Algorithm~\ref{algo:BdryPhase}) that no robot moves to corner during this phase. So, from this phase a robot can never initiate the \textsc{Entry Phase} again and also the corresponding boundary of a robot stays the same during this phase. \textsc{Boundary Phase} only terminates when the configuration becomes an \textsc{InitGather Configuration} or both the robots reach the location of the resource.

In this phase, a non terminated robot say $r$ first checks if it can see any other robot on the grid. If it does not see any other robot on the grid that implies another robot, say $r'$, has already reached the location of $res$. In this case, $r$ simply moves towards the location of the resource $r$ avoiding any corner vertices. When a robot reaches the location of the resource it terminates.
  On the other hand, if $r$ sees another robot $r'$ on the grid then, $r$ first finds out the adjacent vertex on its corresponding boundary which is nearest to the resource. Note that, there can be two such vertices only if $r$ and the resource are on the same line $PD(r)$ but in this case $r$ does not move. Now if $v$ is unique then $r$ calculates the distance from the resource $res$ for the robots $r$ and $r'$ along $BD(r)$ and $BD(r')$ respectively. If both are non zero then, $r$ finds out if $v$ is a corner or not. If $v$ is not a corner and the other robot $r'$ is  not adjacent to some corner on its corresponding boundary, then, $r$ moves to $v$. Otherwise, if $r'$ is adjacent to a corner on its boundary then, $r$ moves to $v$ only if distance of the resource along $BD(r)$ is not equal to one. This technique is required to avoid a livelock scenario. There can be another configuration where distance of the other robot $r'$ along $BD(r')$ is zero but distance of $r$ along $BD(r)$ is strictly greater than one. In this case the robot $r$ simply moves to the vertex $v$. Note that in this case $v$ can not be a corner vertex as $dist(r) >1$. So, during this phase no robot ever moves to a corner and this phase terminates only when both the robot reaches the location of resource or an \textsc{InitGather Configuration} is achieved. Now, we have to ensure that the \textsc{Boundary Phase} terminates within finite rounds. But, before that we need to define a quadrant and proof some results which will be needed to proof the termination of \textsc{Boundary Phase} within $O(T_f \times \max\{m,n\})$ rounds. 

\begin{definition}[Quadrant]
    The grid $G$ is divided into four segments by the two lines $PD(r)$ and $PD(r')$. Each of these segments are called a quadrant. 
\end{definition}

\begin{figure}[ht!]
    \centering
    \includegraphics[width=.6\textwidth]{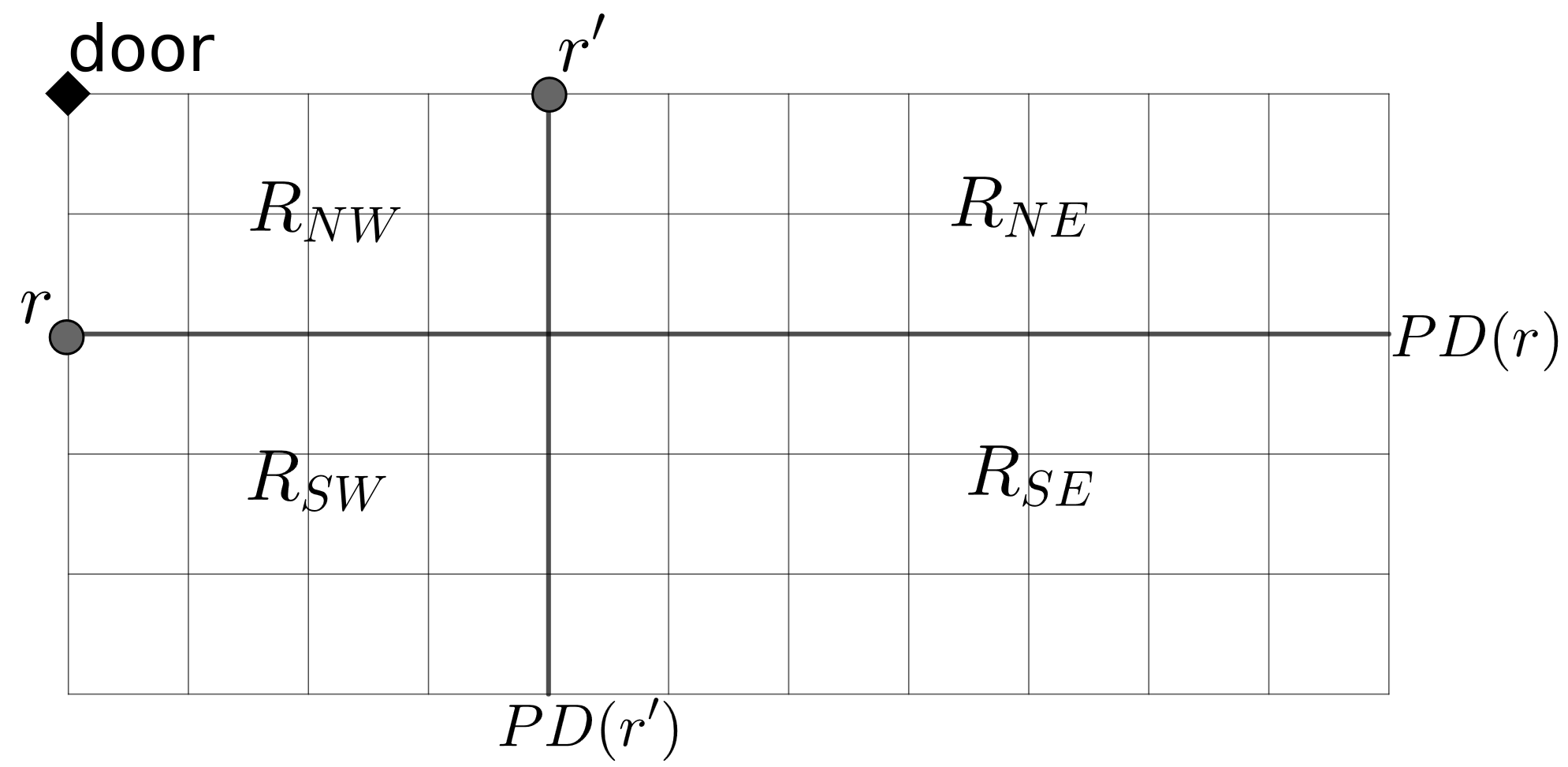}
    \caption{Four quadrants divided by $PD(r)$ and $PD(r')$}
    \label{fig:qudrants}
\end{figure}

 The quadrants on the northeast, northwest, southeast, and southwest are denoted as $R_{NE}, R_{NW}, R_{SE}$ and $R_{SW}$ respectively (Fig.~\ref{fig:qudrants}). Note that the intersection of any two quadrant is a section of either $PD(r)$ or $PD(r')$.
At the beginning of the \textsc{Boundary Phase}, quadrant $R_{NW}$ is a $2 \times 2$ grid, $R_{SW}$ is a $(m-1) \times 2$ grid, $R_{NE}$ is a $2 \times (n-1)$ grid and $R_{SE}$ is a $(m-1) \times (n-1)$ grid. At the beginning of \textsc{Boundary Phase}, the resource $res$ must be either inside or on one of $R_{NE}, R_{SW}$ and $R_{SE}$.    

\begin{algorithm}[H]
\caption{Boundary Phase for robot $r$}\label{algo:BdryPhase}
\eIf{on the same vertex with $res$}
{
    terminate\;
}
{
    \eIf{$r'$ is on the same vertex with $res$}
    {
        move to $res$ along any shortest path avoiding corner in between\;
    }
    {
        $v \leftarrow$ adjacent vertex on $BD(r)$ which is near $res$ along the boundary\;
      \uIf{ $dist(r) \ne 0$ and $dist(r') \ne 0$}
        {
            
            \If{$v$ is not corner}
            {   
                \eIf{$r'$ is not adjacent to a corner}
                {
                     move to $v$\;
                }
                {
                    \If{$dist(r) \ne 1$}
                    {
                        move to $v$\;
                    }
                }
                
            }
        }
        
       \ElseIf{$dist(r') =0$ and $dist(r) >1$}
        {
            Move to $v$\;
        }
          
    }
     
}

\end{algorithm}
\begin{lemma}
    \label{lemma:crossesPDR}
    If \textsc{Boundary Phase} never terminates then, within $O(T_f \times \max\{m, n\})$ rounds, the resource must crosses or moves on to any one of $PD(r)$ or $PD(r')$ at least once.
\end{lemma}
\begin{proof}
    If \textsc{Boundary Phase} never terminates then, no robot ever reaches the resource $res$ and \textsc{InitGather Configuration} is never formed during execution of the \textsc{Boundary Phase}. If possible,  let our claim is false, i.e., the resource $res$ never moves onto and never crosses $PD(r)$ and $PD(r')$ during the execution of \textsc{Boundary Phase}. So $res$ can never be on $PD(r)$ or on $PD(r')$ at the beginning of \textsc{Boundary Phase}. Now there are three cases depending on the location of $res$ at the beginning of the \textsc{Boundary Phase}. The cases are as following:

At the beginning of \textsc{Boundary Phase} the $res$ is on some vertex of,
\begin{itemize}
    \item [I.]$R_{SE} \setminus \{PD(r) \cup PD(r')\}$
    \item [II.]$R_{NE} \setminus \{PD(r) \cup PD(r')\}$
    \item [III.]$R_{SW} \setminus \{PD(r) \cup PD(r')\}$
\end{itemize}

\textit{Case I:} Let $res$ is on some vertex of  $R_{SE} \setminus \{PD(r) \cup PD(r')\}$ at the beginning of the \textsc{Boundary Phase}. According to our assumption $res$ must stay on some vertex of $R_{SE}~\setminus~\{PD(r)~\cup~ PD(r')~\}$ for infinitely many consecutive rounds.

Let at the beginning of some round $t$, the dimension of $R_{SE}$ is $m' \times n'$, where $m',n'~>~2$. then according to algorithm~\ref{algo:BdryPhase}, the dimension of $R_{SE}$ decreases to $(m'-1)\times(n'-1)$ (Fig.~\ref{Fig:L1C11}). So, within $\min\{m-3, n-3\}$ rounds dimension of $R_{SE}$ becomes either $2\times n_1$ or, $m_1 \times 2$, where $ 2\le m_1 < m$ and $2 \le n_1 <n$.

\begin{figure}[h!]
\begin{minipage}[ht]{0.45\linewidth}
\centering
\includegraphics[width=6cm, height=3.5cm]{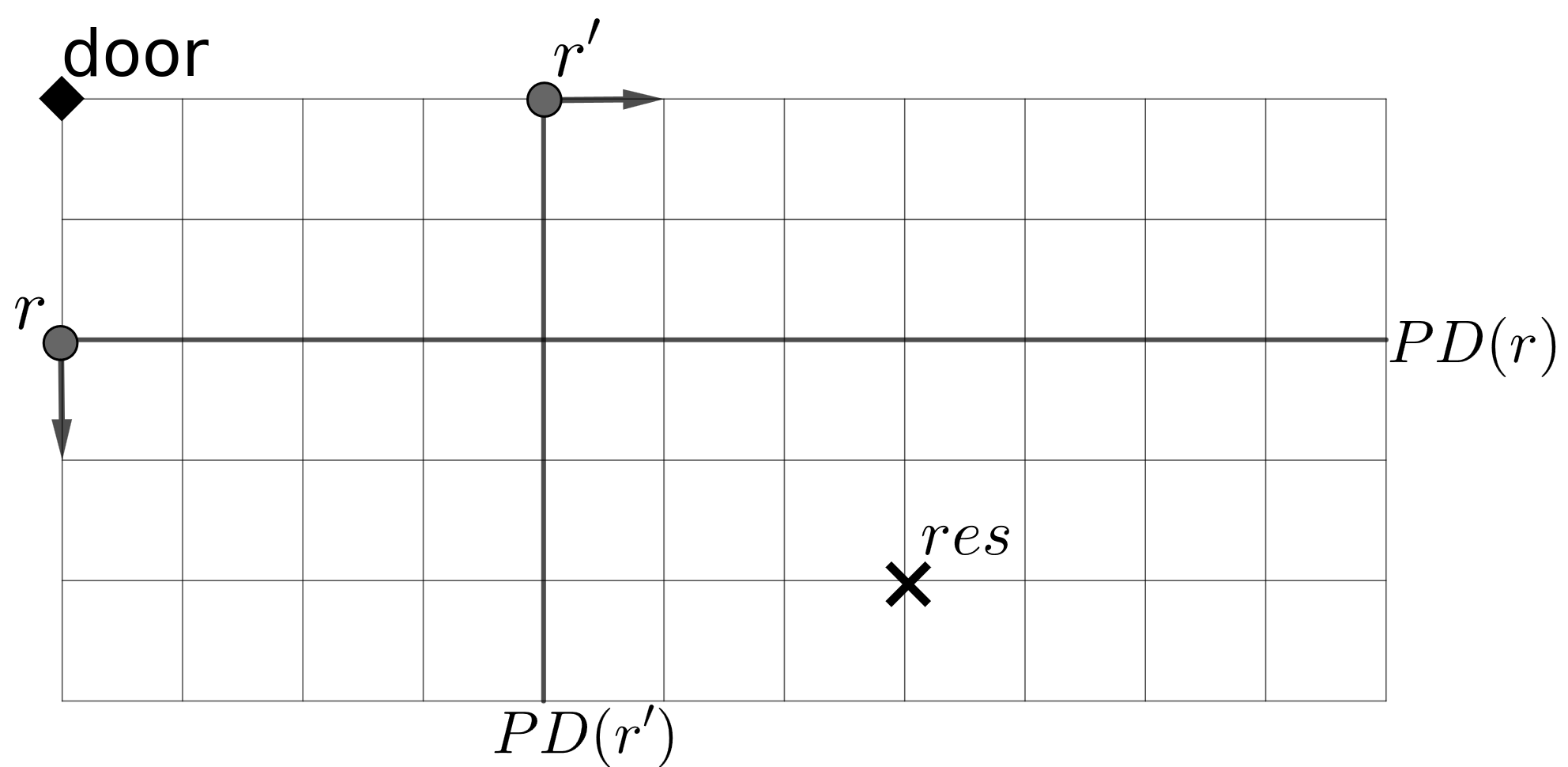}
     \caption{Both height and width of $R_{SE}$ decreases.}
     \label{Fig:L1C11}
\end{minipage}
\hfill
\begin{minipage}[ht]{0.45\linewidth}
\centering
\includegraphics[height=3.5cm, width=6cm]{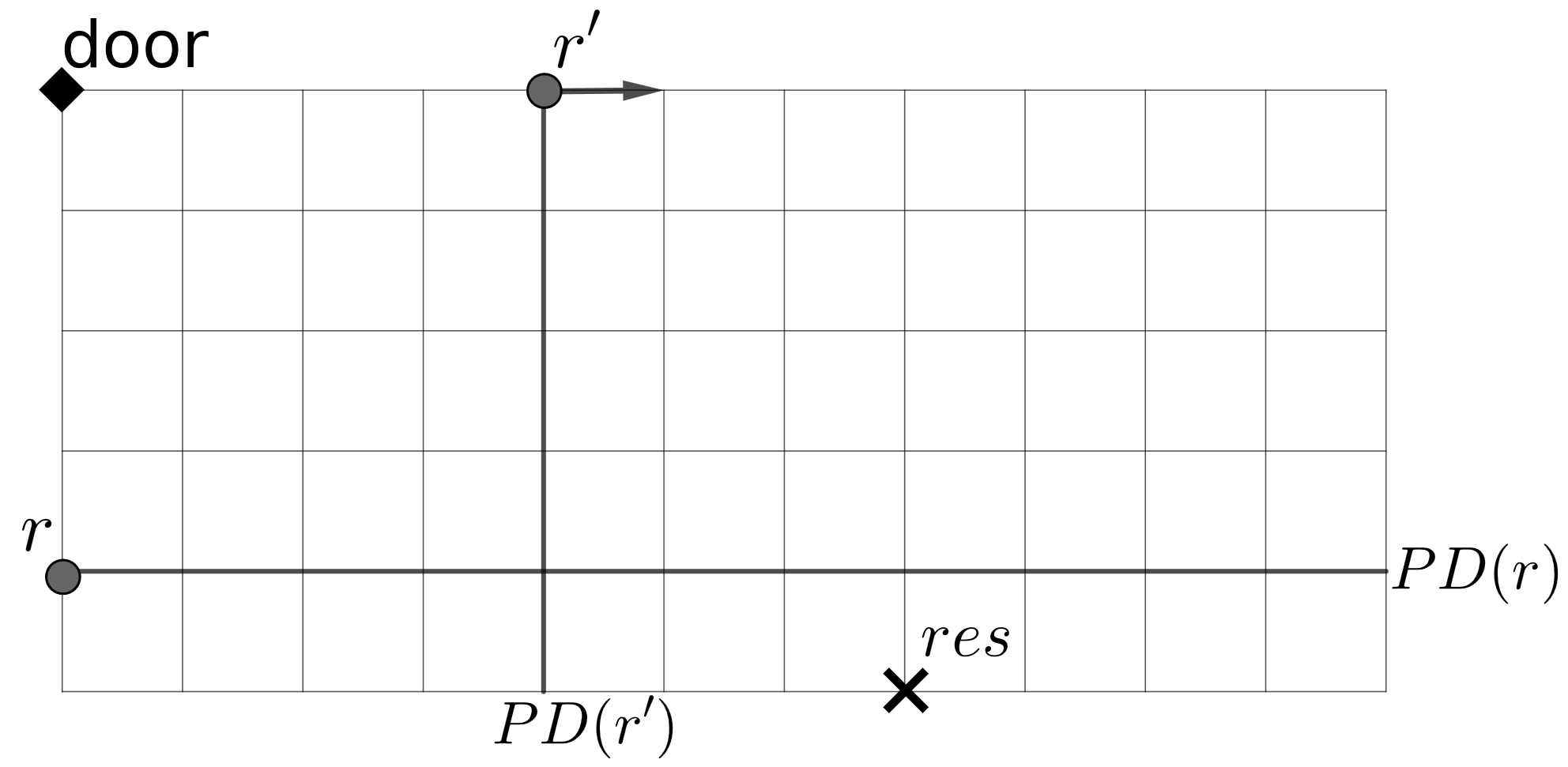}
    \caption{Height of $R_{SE}$ remains same but width decreases.}
    \label{Fig:L1c12}
\end{minipage}
\end{figure}

without loss of generality let, dimension of $R_{SE}$ is $2 \times n_1$ and $2< n_1 < n$ at the beginning of some round $t'$.  In this scenario exactly one robot is adjacent to a corner. Without loss of generality let $r$ be the robot whose adjacent vertex on $BD(r)$ is corner. In this scenario $dist(r) =1$ as $res$ is on $R_{SE}~\setminus~\{PD(r)~\cup~ PD(r')~\}$. Now, if $dist(r') > 1$ then, after completion of round $t'$ dimension of $R_{SE}$ decreases to $2 \times (n_1-1)$ as $r'$ moves towards $res$ along $BD(r')$. On the other hand if  $dist(r') =1$ at the beginning of round $t'$ then, none of $r$ and $r'$ moves according to algorithm~\ref{algo:BdryPhase}. In the worst case, after $T_f$ round during the round $t'+T_f$ the resource $res$ must move. Note that, during this round $res$  can only move parallel to $PD(r)$ and away from $r$ as otherwise the configuration becomes an \textsc{InitGather Configuration} contrary to our assumption. Observe that when $res$ completes the move, $dist(r')$ becomes strictly greater than one and $r'$ moves in the next round and decreases the dimension of $R_{SE}$ to $2 \times(n_1-1)$ (Fig.~\ref{Fig:L1c12}). So, In the worst case at some round $t''$, within\\ $(T_f+1)[\max\{m-3,n-3\}-\min\{m-3,n-3\}] = (T_f+1)[|m-n|]$ rounds the dimension of $R_{SE}$ becomes $2 \times 2$. At the beginning of round $t''+1$ the configuration is as follows (Fig.~\ref{fig:L1C13}),
\begin{itemize}
    \item [1.]$res$ is at the corner which is diagonally opposite to the door vertex.
    \item[2.] $r$ and $r'$ are adjacent to a corner which is not  the door vertex on $BD(r)$ and $BD(r')$ respectively.
\end{itemize}
\begin{figure}[h]
    \centering
    \includegraphics[height=3.5cm,width=6cm]{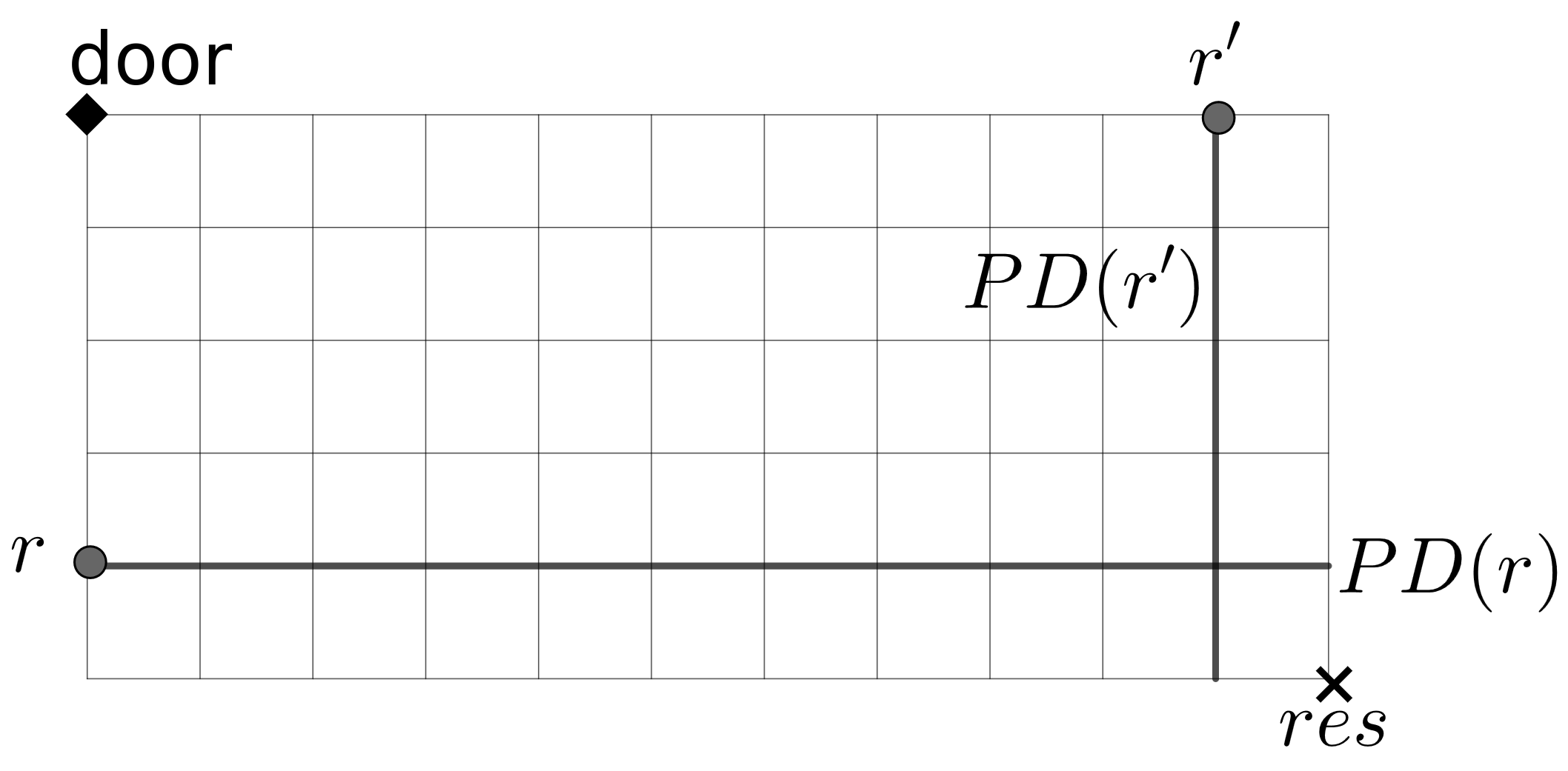}
    \caption{In this configuration $r$ and $r'$ does not move. Whenever $res$ moves the configuration becomes an \textsc{InitGather Configuration}}
    \label{fig:L1C13}
\end{figure}

In this configuration none of $r$ and $r'$ moves and whenever $res$ moves the configuration becomes an \textsc{InitGather Configuration} contrary to our assumption that \textsc{Boundary Phase} never terminates.

Now for the remaining two cases with similar argument it can be shown that the configuration becomes an \textsc{InitGather Configuration} that leads to a contradiction.
So our assumption  that $res$ never crosses or moves onto any of $PD(r)$ and $PD(r')$ was wrong.
Thus, if \textsc{Boundary Phase} doesn't terminate then, $res$ must move onto or crosses any one of $PD(r)$ or $PD(r')$ within $(T_f+1)[|m-n|]+T_f$ rounds in the worst case, which is asymptotically equals to $O(T_f \times \max\{m, n\})$.
\qed
\end{proof}
\begin{lemma}
    \label{lemma:distrlessorequalsone}
    If \textsc{Boundary Phase} never terminates and the $res$ has moved onto or crosses  $PD(R)$ ($R \in \{r, r'\}$) at some round (say $t$), then $dist(R) \le 1$ from round $t$ on wards.
\end{lemma}
\begin{proof}
    If \textsc{Boundary Phase} never terminates then, no robot ever reaches the resource $res$ and \textsc{InitGather Configuration} is never formed during execution of the \textsc{Boundary Phase}. Without loss of generality, let $res$ have crossed or moved onto $PD(r)$ at round $t$. So at the beginning of round $t+1$, $dist(r)$, must be less or equal to one. 

\textit{Case I:} Let $res$ be on $PD(r)$ at the beginning of round $t+1$. Then $dist(r) =0$.
Now during round $t+1$, $res$ either moves along $PD(r)$ (horizontally in Fig.~\ref{fig:qudrants})  or, Perpendicular to $PD(r)$ (vertically in Fig.~\ref{fig:qudrants}) or does not move at all. Now if $res$ moves parallel to $PD(r)$ or does not move at all, then  $dist(r)$ remains the same after completion of round $t+1$ according to the algorithm of \textsc{Boundary Phase}. On the other hand, If $res$ moves Perpendicular to $PD(r)$ during round $t+1$, then $dist(r)$ becomes one after the completion of round $t+1$.

\textit{Case II:} Let $res$ crosses $PD(r)$ at round $t$. Then at the beginning of the round $t+1$, $dist(r) =1$. Now if $res$ moves parallel to $PD(r)$ or does not move at all during round $t+1$ then after the completion of the round, $dist(r)$ either stays one or decreases to zero (Fig.~\ref{Fig:L1C21}, Fig.~\ref{Fig:L1c22}). Now let $res$ moves Perpendicular to $PD(r)$ during round $t+1$, then if $res$ moves towards $PD(r)$ then $dist(r)$ either remains same as one (as $r$ can move along $BD(r)$, $res$ crosses $PD(r)$) or becomes zero in case $r$ does not move along $BD(r)$ (Fig.~\ref{Fig:L1C23}, Fig.~\ref{Fig:L1c24}). On the other hand, if $res$ moves away from $PD(r)$ during round $t+1$, then $dist(r)$ remains one after completion of round $t+1$ as $r$ also moves during round $t+1$ towards the direction of $res$ along $BD(r)$.

So after completion of round $t+1$,\  $dist(r)$ is still less or equal to 1. Now with similar arguments, it is easy to see that if after completion of round $t+i,\  dist(r) \le 1$, then $dist(r) \le 1$ after completion of round $t+i+1$ for some natural number $i$. Hence by Mathematical induction, we can conclude the lemma.
\begin{figure}[h!]
\begin{minipage}[ht]{0.45\linewidth}
\centering
\includegraphics[width=5cm, height=3cm]{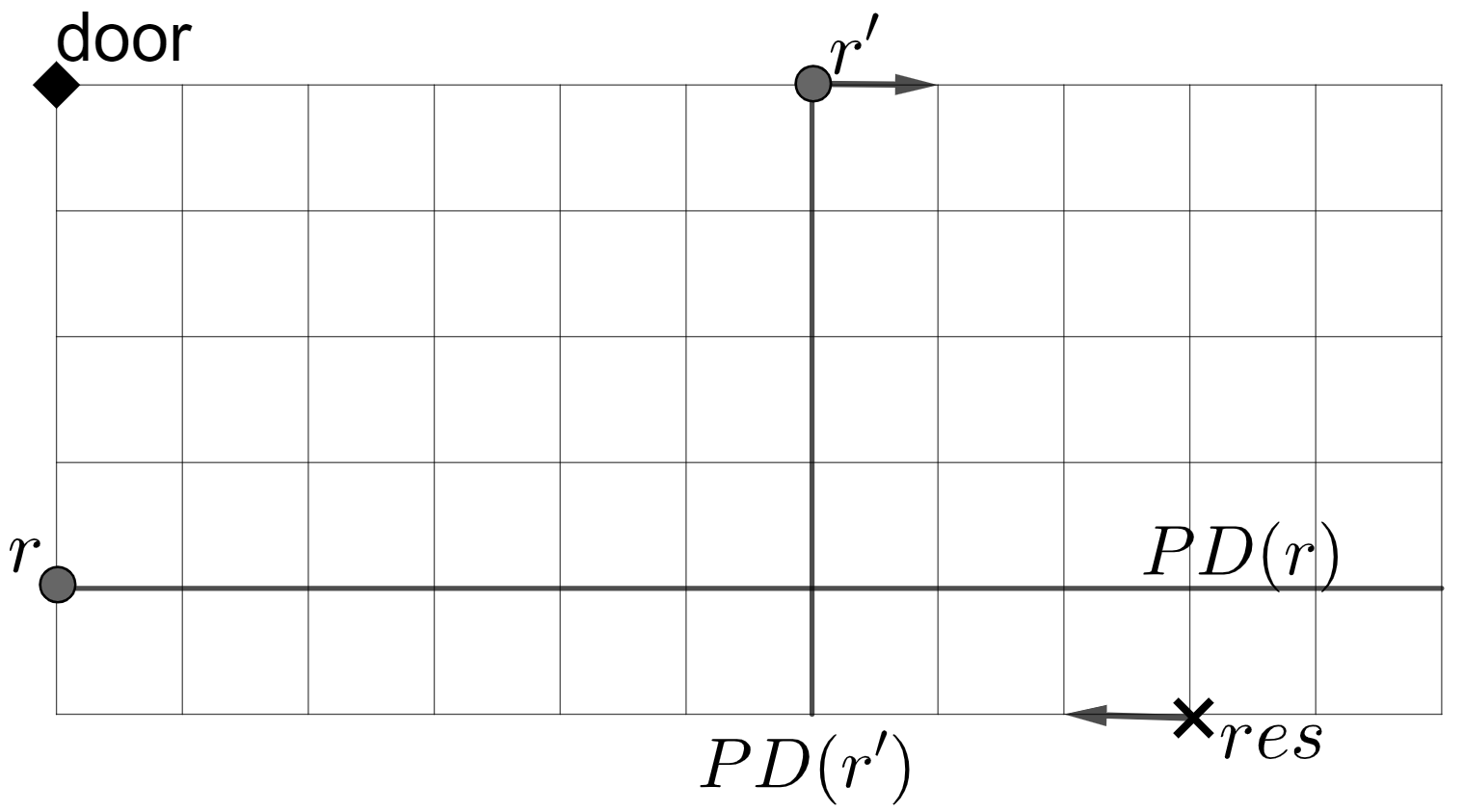}
     \caption{$dist(r)$ remains one.}
     \label{Fig:L1C21}
\end{minipage}
\hfill
\begin{minipage}[ht]{0.45\linewidth}
\centering
\includegraphics[height=3cm, width=5cm]{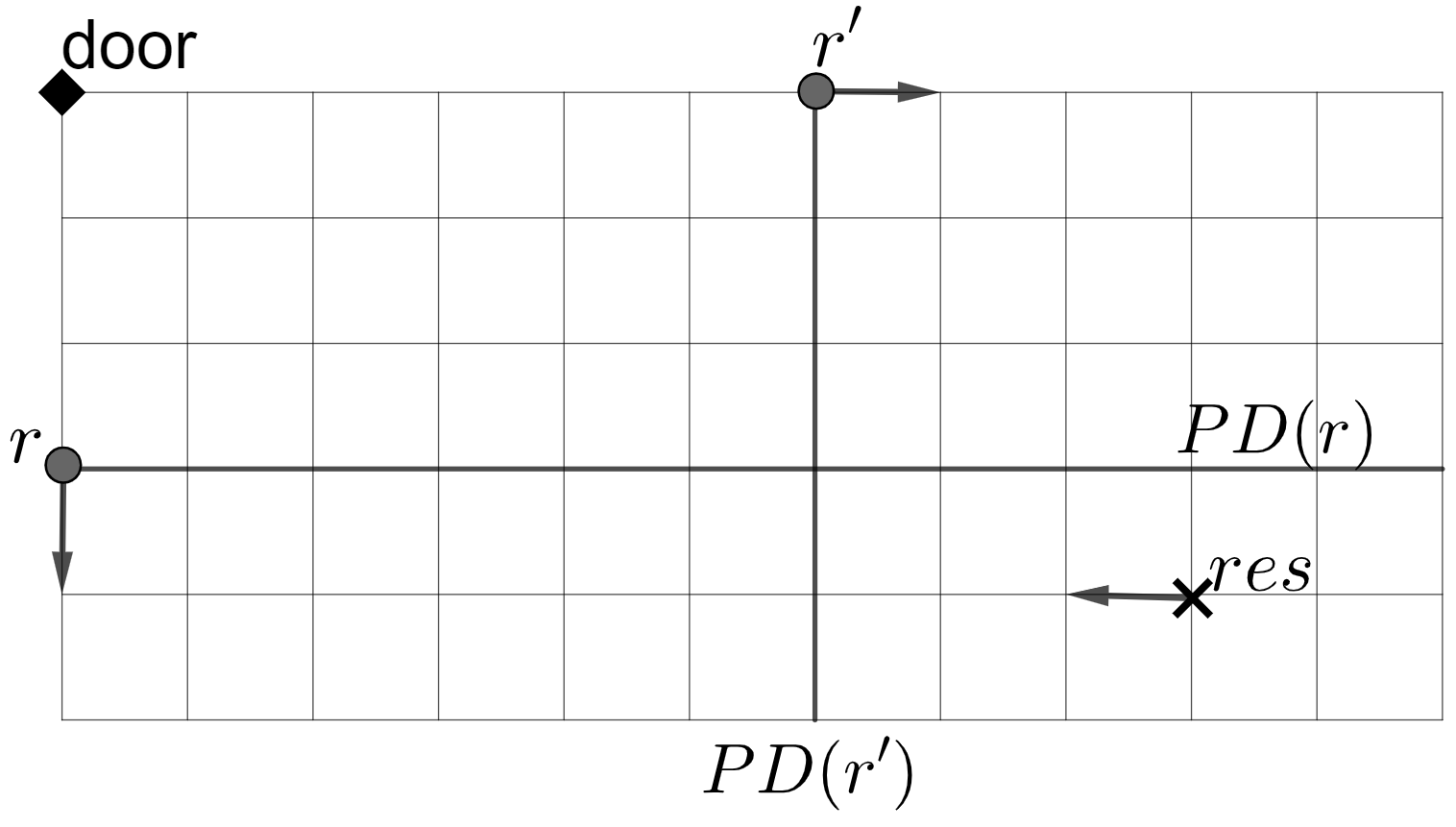}
    \caption{$dist(r)$ becomes zero.}
    \label{Fig:L1c22}
\end{minipage}
\begin{minipage}[ht]{0.45\linewidth}
\centering
\includegraphics[width=5cm, height=3cm]{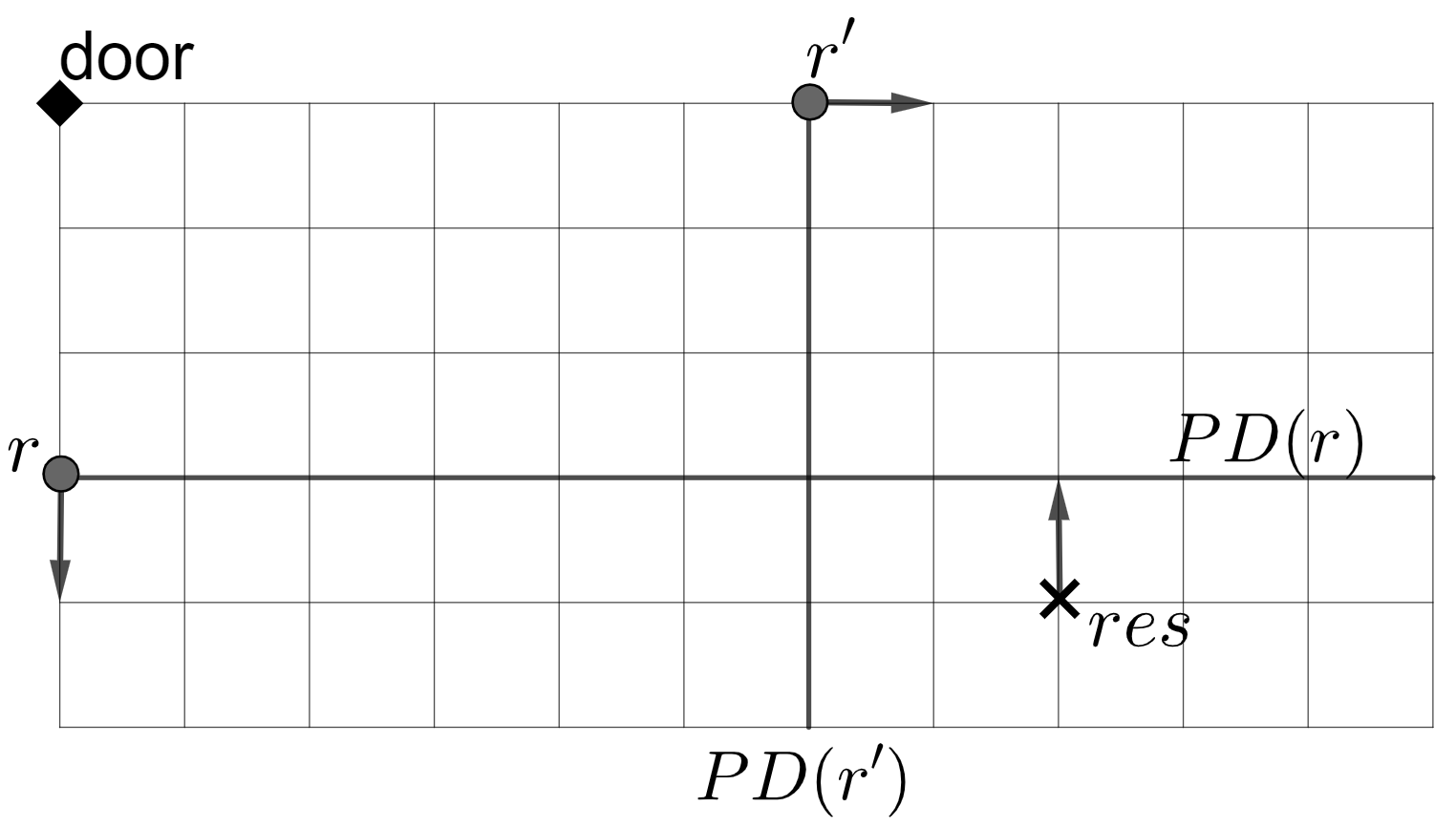}
     \caption{$dist(r)$ remains one as $res$ crosses $PD(r)$.}
     \label{Fig:L1C23}
\end{minipage}
\hfill
\begin{minipage}[ht]{0.45\linewidth}
\centering
\includegraphics[height=3cm, width=5cm]{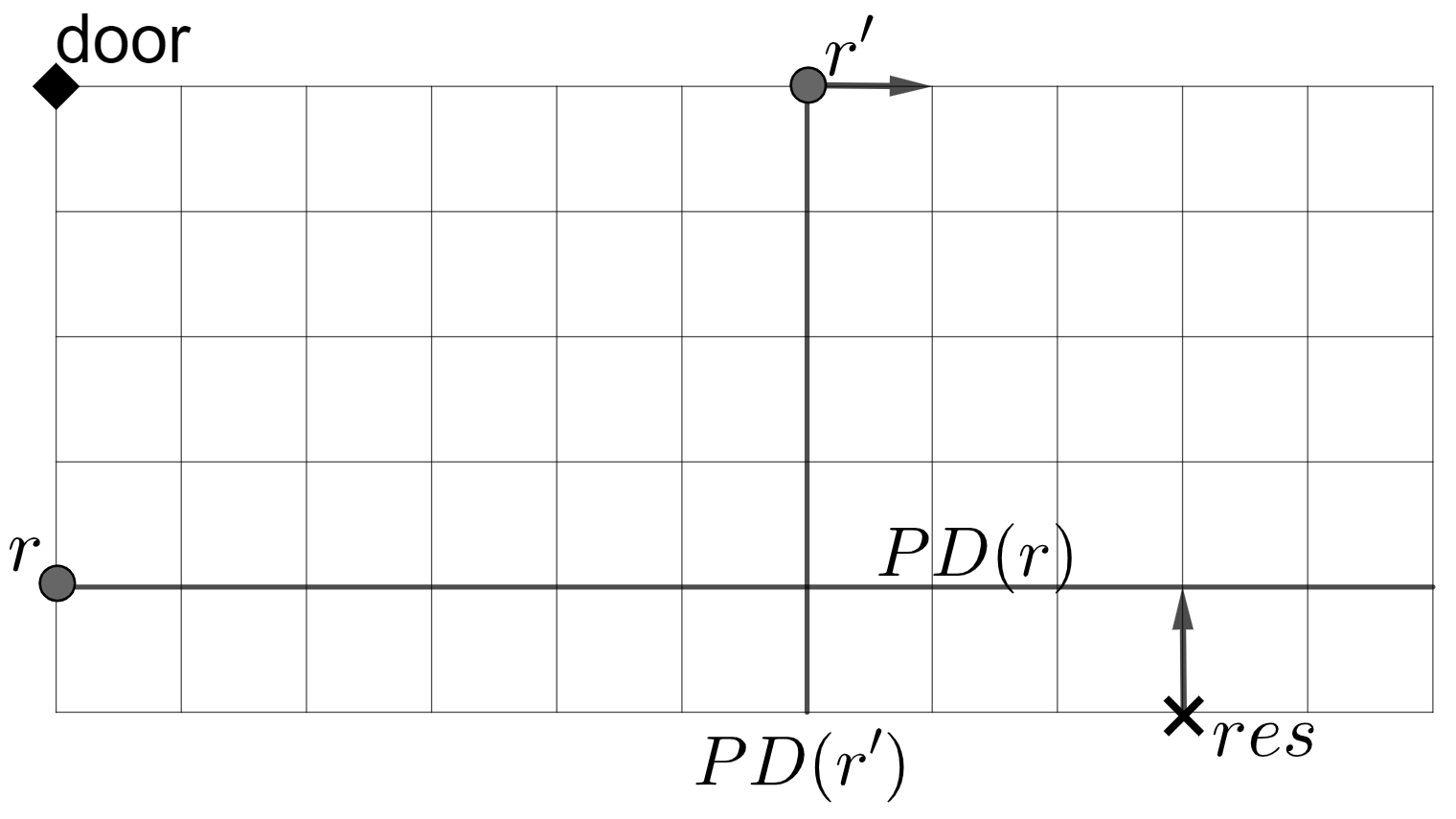}
    \caption{$dist(r)$ becomes zero as $r$ does not move.}
    \label{Fig:L1c24}
\end{minipage}
\end{figure}
\qed
\end{proof}
\begin{lemma}
    \label{lemma:nevercrossesPDR'}
    If \textsc{Boundary Phase} never terminates and $res$ has moved onto or crossed $PD(R)$ where $R \in \{r, r'\}$ at some round $t$, then $res$ never crosses $PD(R')$ from round $t$ on wards (Here $R' = r $ if $R = r'$ and $R' = r'$ if $R = r$).
\end{lemma}
\begin{proof}
    If \textsc{Boundary Phase} never terminates then, no robot ever reaches the resource $res$ and \textsc{InitGather Configuration} is never formed during execution of the \textsc{Boundary Phase}. Without loss of generality let $res$ move onto or crosses $PD(r)$ at some round $t_0$ and $BD(r')$ is the boundary of $G$ at the north (Fig.~\ref{fig:qudrants}). By Lemma~\ref{lemma:distrlessorequalsone} ,for any round $t > t_0$, $dist(r)$ remains less or equals to one. We have to show that from round $t$ on wards $res$ never crosses or moves onto $PD(r')$. If possible let $res$ moves onto or crosses $PD(r')$ at some round $t_1 > t_0$. Then $res$ must moves parallel to $PD(r)$ or doesn't move at all during round $t_1$.

Note that, if $dist(r) = 0$ at the beginning of round $t_1$ then, after completion of the round $dist(r)$ remains zero and $dist(r')$ becomes less than or equals to one. So, after completion of round $t_1$ the configuration becomes an \textsc{InitGather Configuration} contrary to our assumption. So,let $dist(r)=1$ at the beginning of round $t_1$. Now, after completion of the round, if $res$ moves onto $PD(r')$ i.e., $dist(r')$ becomes zero then, the configuration again becomes an \textsc{InitGather Configuration} as $dist(r)$ either remains one or becomes zero. So, $res$ can only cross $PD(r')$ during  round $t_1$. After crossing $PD(r')$, $dist(r')$ becomes one. Now, for $res$ to cross $PD(r')$ without achieving an \textsc{InitGather Configuration} is possible only when $r$ does not move otherwise $dist(r)$ becomes zero and we arrive at an \textsc{InitGather Configuration} contrary to the assumption. Observe that, $r$ does not move during round $t_1$ only if it is located adjacent to a corner on $BD(r)$. Now at the beginning of round $t_1$, $dist(r')$ is also one (as it is assumed that it crosses $PD(r')$ during round $t_1$) (Fig.~\ref{fig:L1C31}). In this scenario $r'$ doesn't move during the round $t_1$ according to algorithm~\ref{algo:BdryPhase}. So, $res$ can't cross $PD(r')$ without moving onto it. As described earlier this leads to an \textsc{InitGather Configuration} which is a contradiction. So, for \textsc{Boundary Phase} to never terminate, if $res$ crosses $PD(R)$ during round $t_0$ then from round $t_0$ on wards it never crosses or moves onto $PD(R')$. 

\begin{figure}
    \centering
    \includegraphics[height=3.5cm,width=6cm]{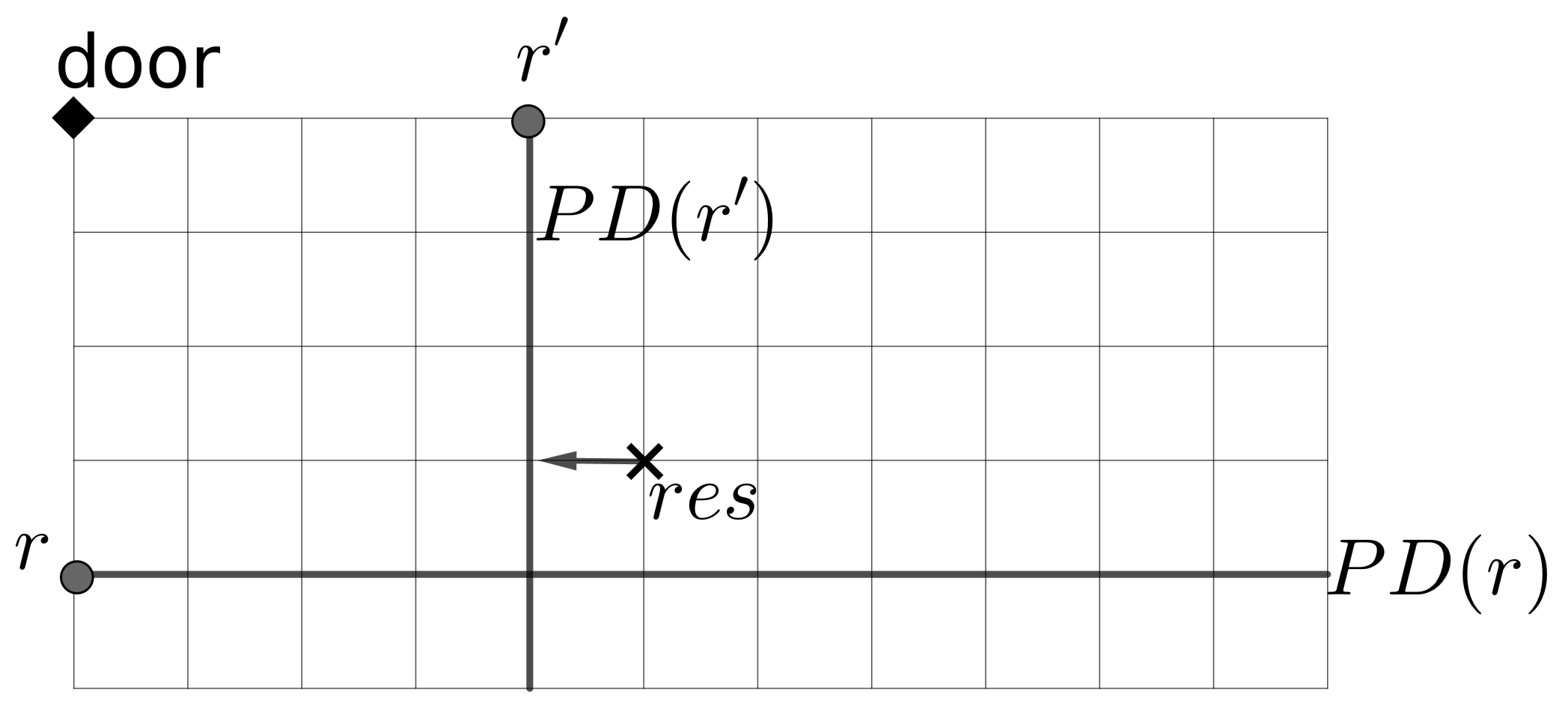}
    \caption{$res$ moves onto $PD(r')$ and creates an \textsc{InitGather Configuration}.}
    \label{fig:L1C31}
\end{figure}
\qed
\end{proof}

Now using the above lemmas we prove the following theorem.

\begin{theorem}
    \label{Thm:bdryPhsTrmintn}
     For a grid of dimension $m\times n$, the \textsc{Boundary Phase} terminates within $O(T_f \times \max\{m, n\})$ rounds.
\end{theorem}
\begin{proof}

If possible, let us assume \textsc{Boundary Phase}  never terminates. This implies no robot ever moves to the location of $res$ and \textsc{InitGather Configuration} is never formed during the execution of \textsc{Boundary Phase}.
From the above three lemmas (i.e., Lemma~\ref{lemma:crossesPDR}, Lemma~\ref{lemma:distrlessorequalsone} and Lemma~\ref{lemma:nevercrossesPDR'}) it can be said that, $res$ must crosses or moves onto  $PD(R)$ where $R \in \{r, r'\}$ at some round $t$ for the first time ,within $T_f \times O(\max\{m, n\})$ rounds. Then from round $t$ on wards, $dist(R)$ always remains less or equals to one and $res$ never crosses $PD(R')$  where $R' =r$ if $R = r'$ and $R' =r'$ if $R = r$. 

Now without loss of generality let, $res$ has moved onto or crossed $PD(r)$ at some round $t$ for the first time. Also, let, $BD(r')$ is the boundary on the north of the grid (Fig.~\ref{fig:qudrants}). Then from round $t$ on wards $res$ must lie inside $(R_{NE} \cup R_{SE}) \setminus PD(r')$ (i.e., $dist(r') \ge 1$ from round $t$ on wards). So, in the worst case, $r'$ must move away from the door along $BD(r')$ in each  $T_f$ consecutive rounds once, from round $t$ on wards. Let after completion of round $t$ the dimension of $(R_{NE} \cup R_{SE})$ is $m' \times n'$, where the length $n' < n$. Note that after  $(n'-2)\times(T_f+1)$ rounds $(R_{NE} \cup R_{SE}) \setminus PD(r')$ is the boundary of the grid, say $BD_{par}(r)$,  which is parallel to $BD(r)$ (Fig.~\ref{fig:L1f}). In this scenario $res$ must be on some vertex of $BD_{par}(r)$. Note that in this configuration $r'$ is adjacent to a corner so it does not move and $res$ can not be on $PD(r)$ as otherwise it would be a \textsc{InitGather Configuration} and \textsc{Boundary Phase} terminates. Now in this configuration $res$ can be either on some vertex of  $R_{NE}\setminus\{PD(r)\cup PD(r')\}$ or on some vertex of $R_{SE}\setminus\{PD(r)\cup PD(r')\}$ and also remains so in the upcoming rounds. Let without loss of generality $res$ is on some vertex of  $R_{NE}\setminus\{PD(r)\cup PD(r')\}$ at the beginning of some round $t$. Let dimension of $R_{NE}$ is $m'\times 2$ at the beginning of round $t$ (Fig.~\ref{fig:L1f}). Then by the similar argument used to prove \textit{Case I.} of Lemma~\ref{lemma:crossesPDR} we can conclude that if no robot reaches the location of $res$, within at most $(T_f+1)\times (m'-2)+T_f$ more rounds the configuration becomes an \textsc{InitGather Configuration}, which is a contradiction. Hence our assumption that \textsc{Boundary Phase} never terminates is wrong. And as calculated the \textsc{Boundary Phase} terminates in $O(T_f\times\max\{m,n\})+(T_f+1)\times(m+n-6)+T_f$ rounds in the worst case which is asymptotically equals to $O(T_f\times\max\{m,n\})$.

\begin{figure}
    \centering
    \includegraphics[height=3cm,width=6cm]{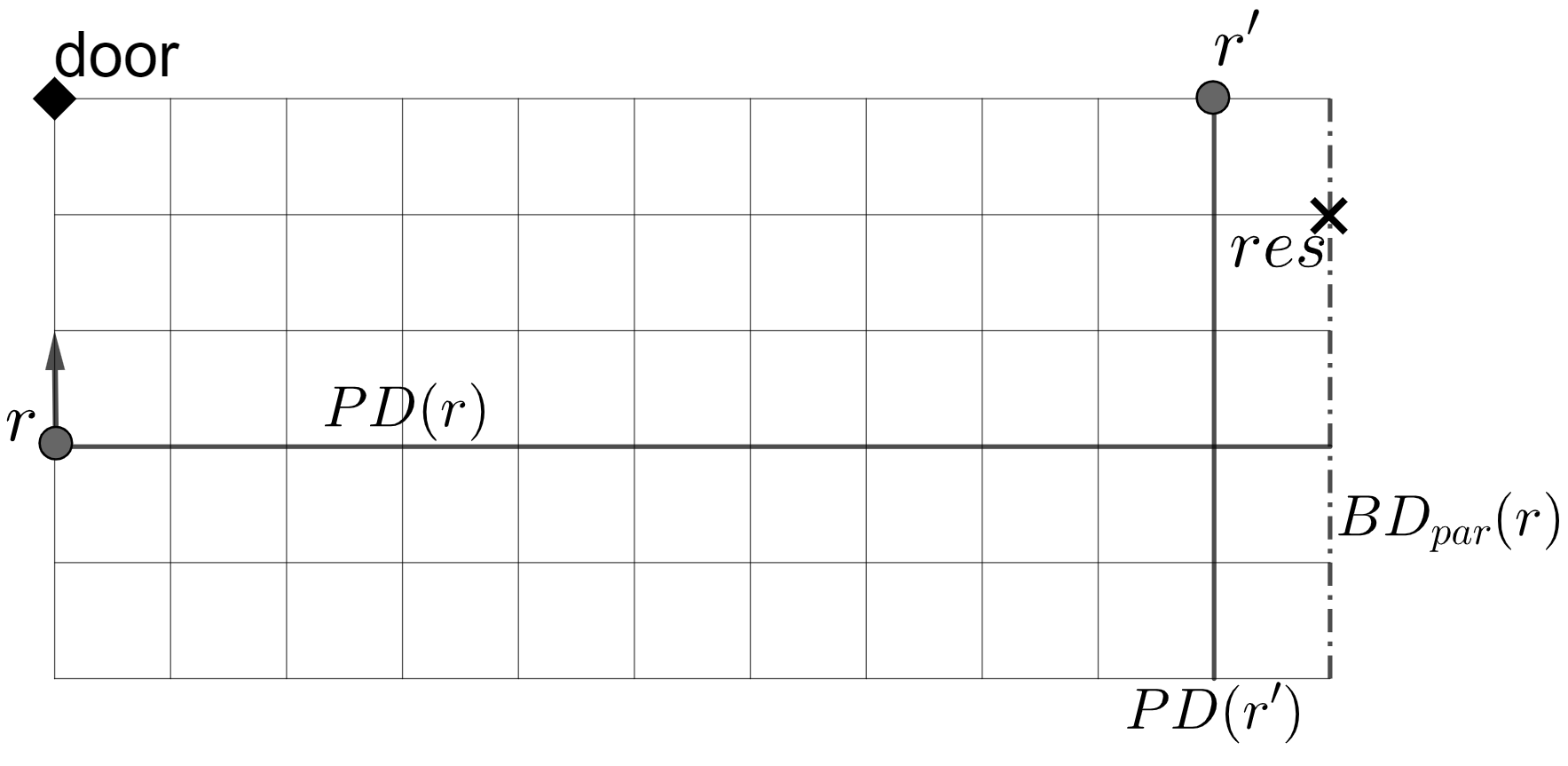}
    \caption{dimension of $R_{NE}$ decreases from $m' \times 2$ to $(m'-1)\times 2$}
    \label{fig:L1f}
\end{figure}
\qed
\end{proof}

\subsection{\textsc{Gather Phase}}
\textsc{Gather Phase} starts if none of the two robots reaches the location of $res$  after the termination of the \textsc{Boundary Phase}. Throughout the execution of this phase, the configuration will remain an \textsc{InitGather Configuration} (Lemma~\ref{Lemma:InitGAther}). So, in each round, a robot will lie on the same line (say $L$) along with $res$, and the perpendicular distance of the other robot to the line passing through $res$ and perpendicular to $L$ must be at most one.  
During this phase, if a robot is in the same location with $res$, it terminates and the other robot moves to the location of $res$ along any shortest path. On the other hand when none of the robots are on the same vertex with $res$, the robot on $L$ checks if $res$ is adjacent to it. If $res$ is not on its adjacent vertex, it moves towards $res$ along $L$. Otherwise, if $res$ is on its adjacent vertex then it moves towards $res$ along $L$  only if it sees $res$ is on a corner and the other robot, say $r'$, is also on another adjacent vertex of $res$. Now if the robot is not on any line along with the resource $res$, (i.e., perpendicular distance of the robot to the line through $res$ and perpendicular to $L$ is one) then, the robot will move parallel to $L$ towards $res$. The pseudo code of \textsc{Gather Phase} is as follows.

\begin{algorithm}[H]
\caption{Gather Phase for robot $r$}\label{algo:GathrPhase}
\eIf{$r$ is on same vertex with $res$}
{
    terminate\;
}
{
    \eIf{$r'$ is on the same vertex with $res$}
    {
        move to $res$ along any shortest path avoiding door vertex\;
    }
    {
        \eIf{$r$ is on a line $L$ with $res$}
        {
            \eIf{$res$ is not adjacent to $r$}
            {
                move towards $res$ along $L$\;
            }
            {
                \If{$res$ is at a corner and $r'$ is adjacent to $res$}
                {
                    move towards $res$ along $L$\;
                }
            }
        }
        {
            move parallel to $L$ towards $res$\;
        }
    }
}

\end{algorithm}
 Before proving the correctness of the \textsc{Gather Phase}, let us discuss some notations used in the following proofs.

 Let, $r$ be the robot which is on the same line along with the resource $res$ in an \textsc{InitGather Configuration}. This line is denoted as $L$. We denote the line passing through $res$ and perpendicular to $L$ as $L'$. Now by $L^1$ ans $L^{-1}$ we denote the lines parallel to $L'$ and one hop distance apart from $L'$. $L^1$ is the line furthest to $r$ compared to $L^{-1}$. Observe that $r'$ can be on any one of $L^{-1}, L', L^1$ in an \textsc{InitGather Configuration}.
Now we can prove the following lemma.

 \begin{lemma}
     \label{lemma:notonSameLine}
     By executing Dynamic Rendezvous algorithm two robots $r$ and $r'$ never moves to same line from an \textsc{InitGather Configuration}.
 \end{lemma}
 \begin{proof}
     Let at the beginning of some round $t$, the configuration is an \textsc{InitGather Configuration}. Observe that $r$ and $r'$ are not on the same line at the beginning of the round $t$. If possible let after completion of round $t$, $r$ and $r'$ moves to the same line. Without loss of generality let, $r$ be the robot on the same line $L$ along with $res$. Then, $r'$ can be on any one of $L^{-1}, L', L^1$ at the beginning of round $t$. Now, we have three cases depending on the location of $r'$.

     \textit{Case I:} Let $r'$ is on $L'$ at the beginning of round $t$. Then for being on the same line after the completion of the round either $r$ moves to $L'$ or, $r'$ moves to $L$. Without loss of generality, let $r$ moves to $L'$ during round $t$. This implies at the beginning of round $t$, $r$ is adjacent to $res$ on $L$ (Fig.~\ref{Fig:GL1C1}). Now according to the algorithm~\ref{algo:GathrPhase}, $r$ doesn't move during round $t$, a contradiction.

     \textit{Case II:} Let $r'$ is on $L^1$ at the beginning of round $t$. Then during round $t$, $r'$ moves parallel to $L$ to $L'$. So to be on the same line after completion of the round, $r$ must move onto $L'$ during this round (Fig.~\ref{Fig:GL1C2}). But for the same reason used in \textit{Case I}, $r$ can not move during round $t$ which leads to a contradiction again.

     \textit{Case III:} Let $r'$ is on $L^{-1}$ at the beginning of round $t$. According to algorithm~\ref{algo:GathrPhase}, $r'$ moves to $L'$ during this round. So for being on the same line $r$ must move to $L'$ too during round $t$. For this to happen, $r$ must be on the line $L^{-1}$ along with $r'$ at the beginning of round $t$ (Fig.~\ref{Fig:GL1C3}). But that is a contradiction as at the beginning of round $t$ the configuration is an \textsc{InitGather Configuration}.

     For all of this cases we reach a contradiction assuming that after completion of round $t$, $r$ and $r'$ moves to the same line. Hence the lemma.
     \qed
 \end{proof}

 \begin{figure}[h!]
\begin{minipage}[ht]{0.45\linewidth}
\centering
\includegraphics[width=4.5cm, height=3 cm]{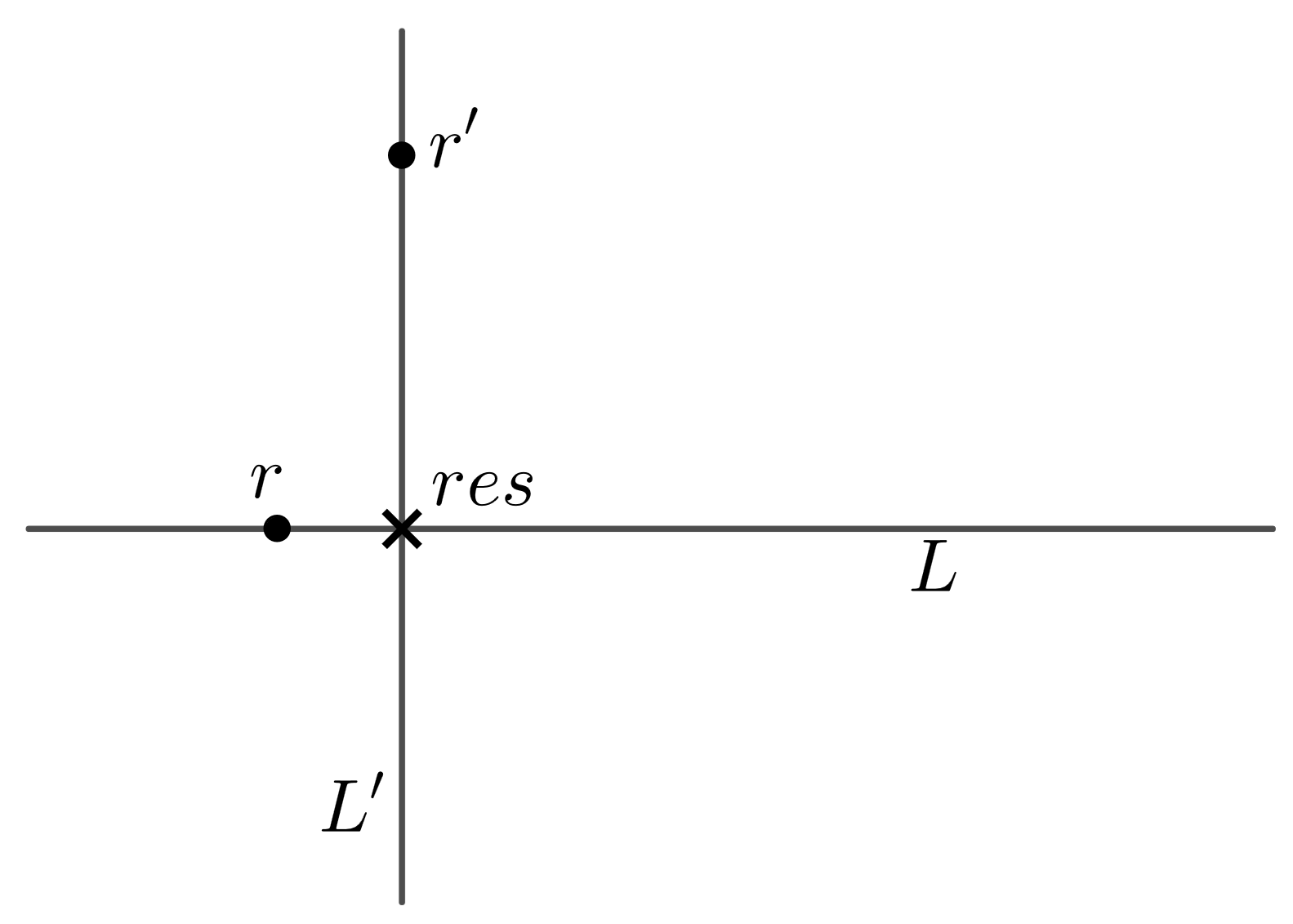}
     \caption{$r'$ is on $L'$. $r$ doesn't move to $L'$.}
     \label{Fig:GL1C1}
\end{minipage}
\hfill
\begin{minipage}[ht]{0.45\linewidth}
\centering
\includegraphics[height=3cm, width=4.5cm]{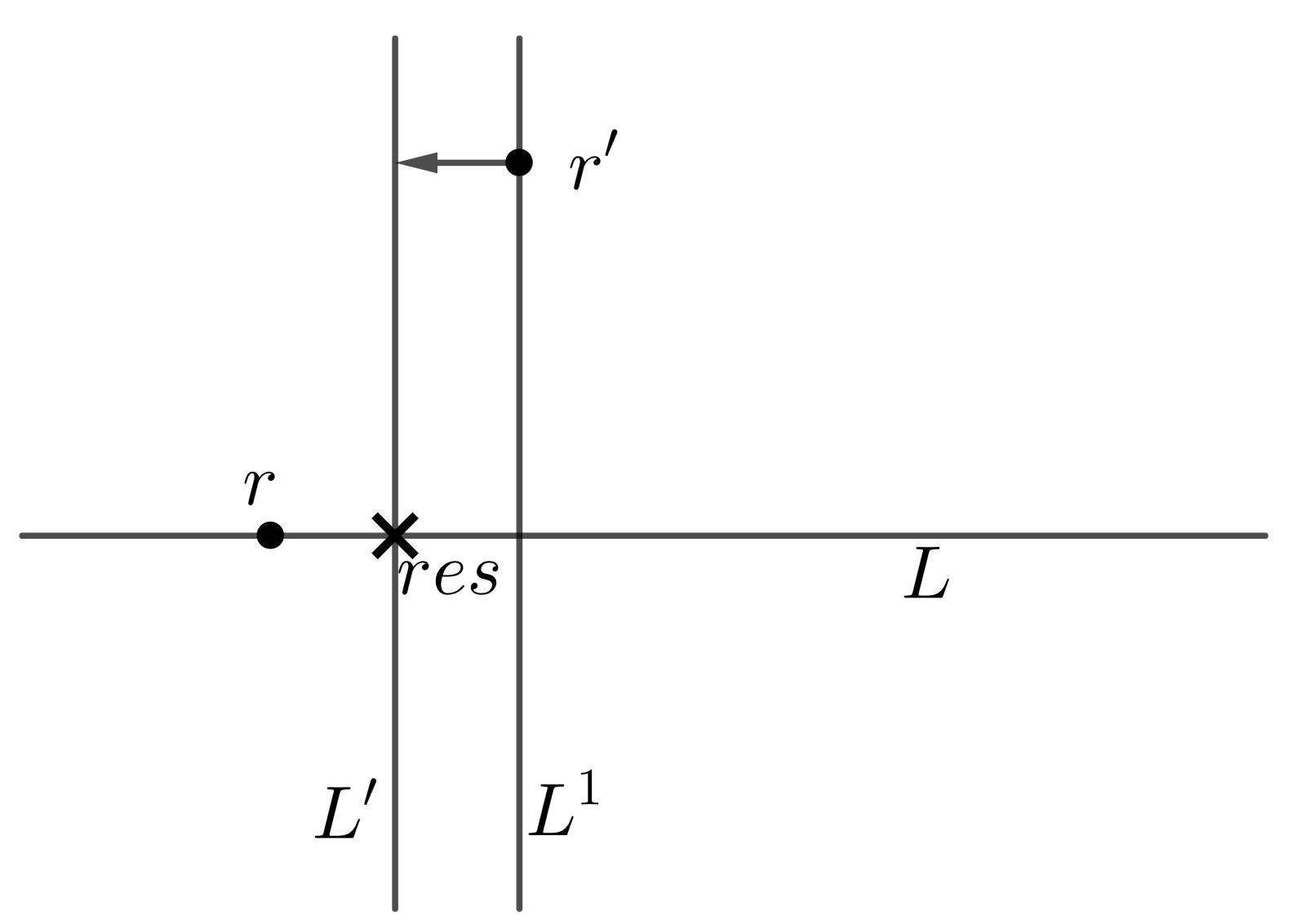}
    \caption{$r'$ is on $L^1$.$r$ doesn't move to $L'$.}
    \label{Fig:GL1C2}
\end{minipage}
\hfill
\begin{minipage}[ht]{\linewidth}
\centering
\includegraphics[height=3cm, width=4.5cm]{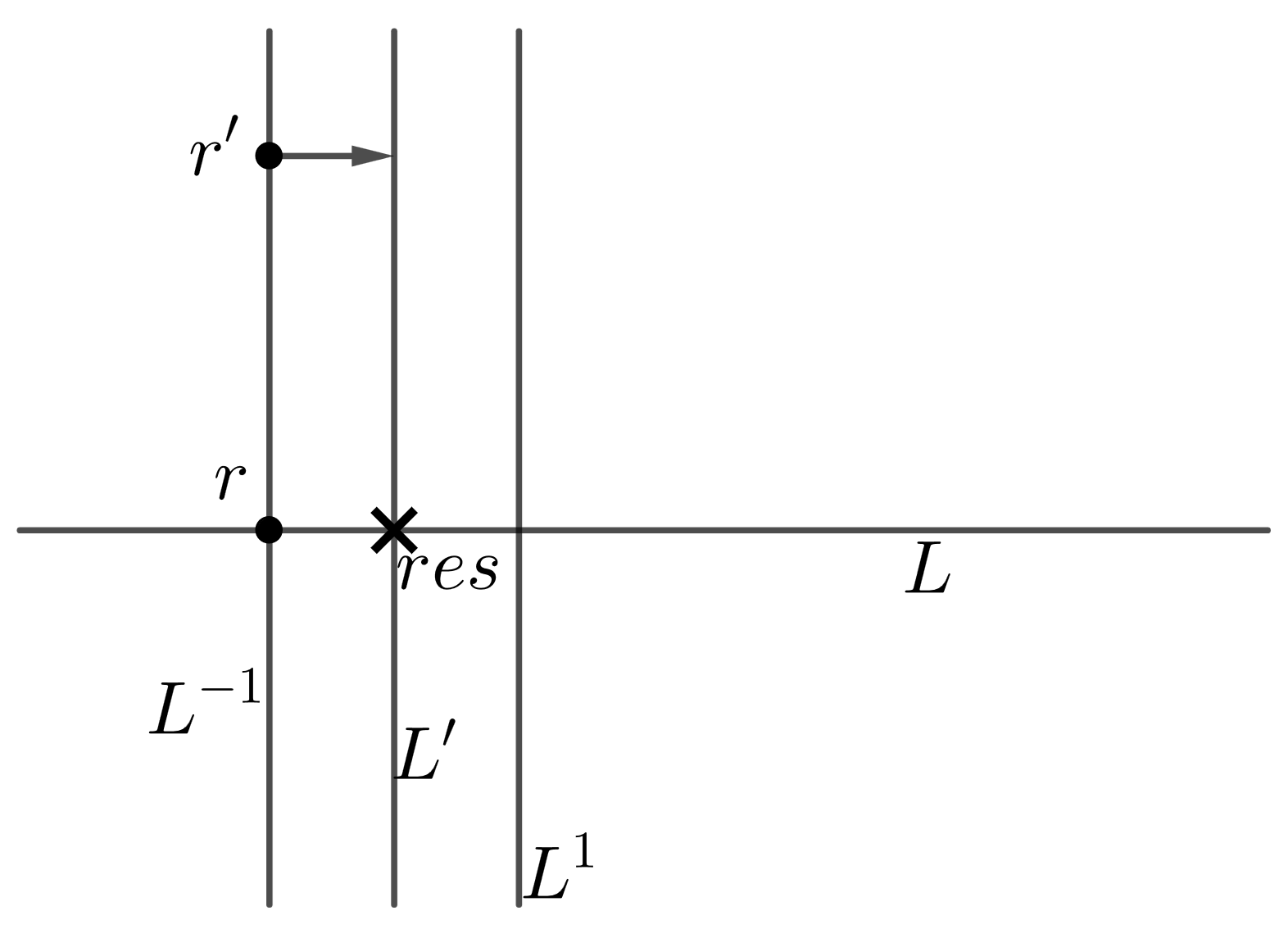}
    \caption{This configuration is not possible.}
    \label{Fig:GL1C3}
\end{minipage}
\end{figure}

Next we have the following lemma.
 \begin{lemma}
     \label{Lemma:InitGAther}
     If no robots terminate, an \textsc{InitGather Configuration} remains an \textsc{InitGather Configuration} after one execution of the algorithm Dynamic Rendezvous. 
 \end{lemma}
 \begin{proof}
Let at the beginning of some round $t$ the configuration is an \textsc{InitGather Configuration}. Let $r$ be the robot on the same line $L$ along with $res$. Now, the other robot $r'$ can be on any one of $L^{-1}, L'$ and $L^1$ at the beginning of round $t$. By the Lemma~\ref{lemma:notonSameLine}, it can be made sure that, $r$ and $r'$ does not move onto same line after completion of round $t$. So, to prove this lemma it is sufficient to show that, after completion of round $t$, one of $r$ and $r'$ is on the same line, say $L_d$, along with $res$ and the perpendicular distance of the other robot to the line $L_d'$ is at most one where, $L_d'$ is the line perpendicular to $L_d$ and passing through $res$. Now based on the movement of $res$ during round $t$, we have three cases as following.
\begin{itemize}
    \item [I.] $res$ moves along $L$.
    \item [II.] $res$ moves along $L'$.
    \item [III.] $res$ does not move.
\end{itemize}

\textit{Case I:} Let $res$ moves along line $L$ during round $t$. Now even if $r$ moves in round $t$, it stays on $L$ according to algorithm~\ref{algo:GathrPhase}. Now, let $r'$ can be on $L'$or, $L^1$ or, $L^{-1}$  at the beginning of round $t$. Irrespective of the position, $r'$ reaches either on $L^1$ or on $L^{-1}$ in the new configuration after completion of the round $t$ (Fig.~\ref{Fig:Lemma1pic1case1} and Fig.~\ref{Fig:lemma1pic2case1}). Thus for this case the configuration remains an \textsc{InitGather Configuration}.

 \begin{figure}[h!]
\begin{minipage}[ht]{0.45\linewidth}
\centering
\includegraphics[width=6 cm, height=4 cm]{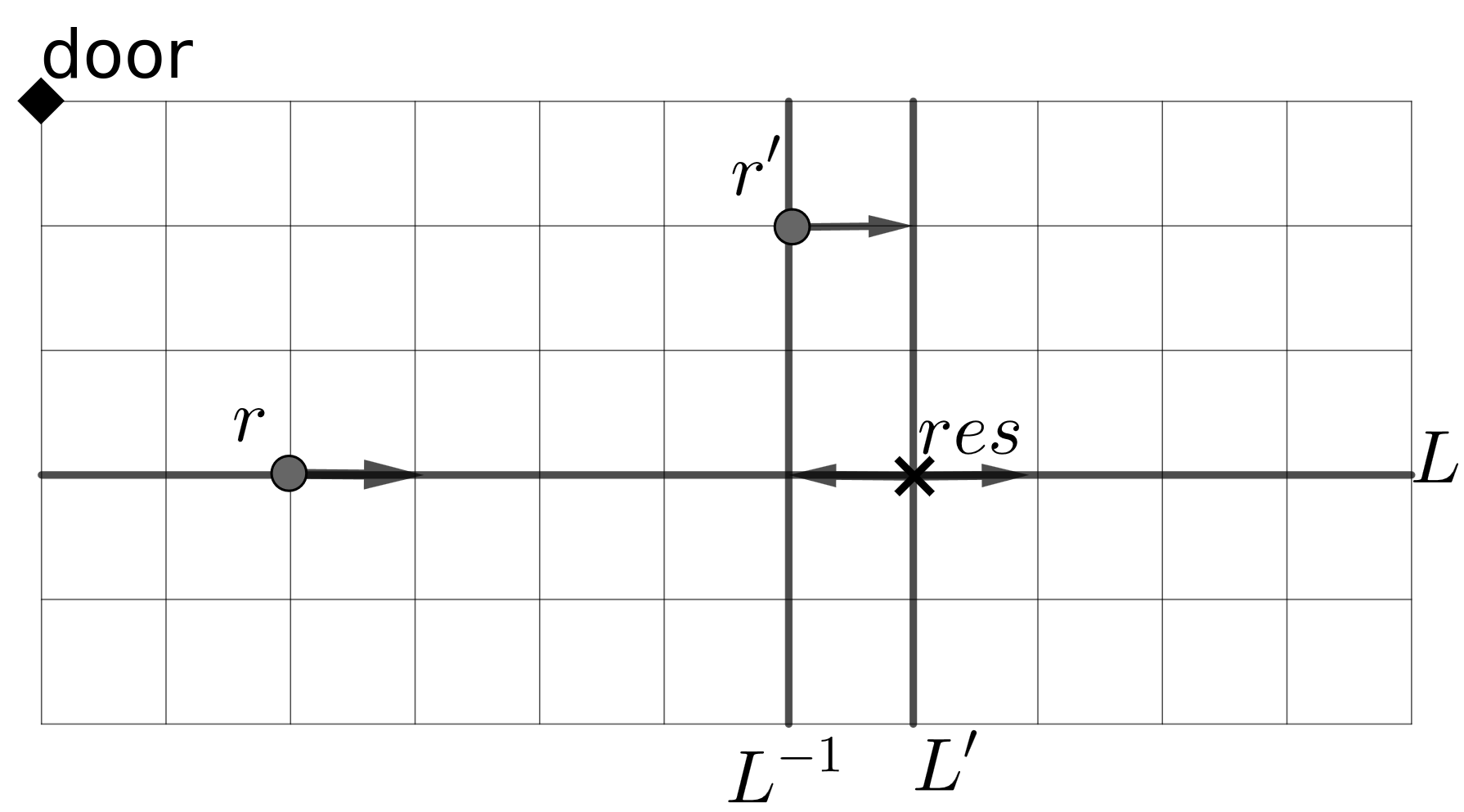}
     \caption{$r'$ moves to $L'$ and $res$ moves to either $L^{-1}$ or $L^1$ during round $t$. $r$ stays on $L$ along with $res$.}
     \label{Fig:Lemma1pic1case1}
\end{minipage}
\hfill
\begin{minipage}[ht]{0.45\linewidth}
\centering
\includegraphics[height=4cm, width=6cm]{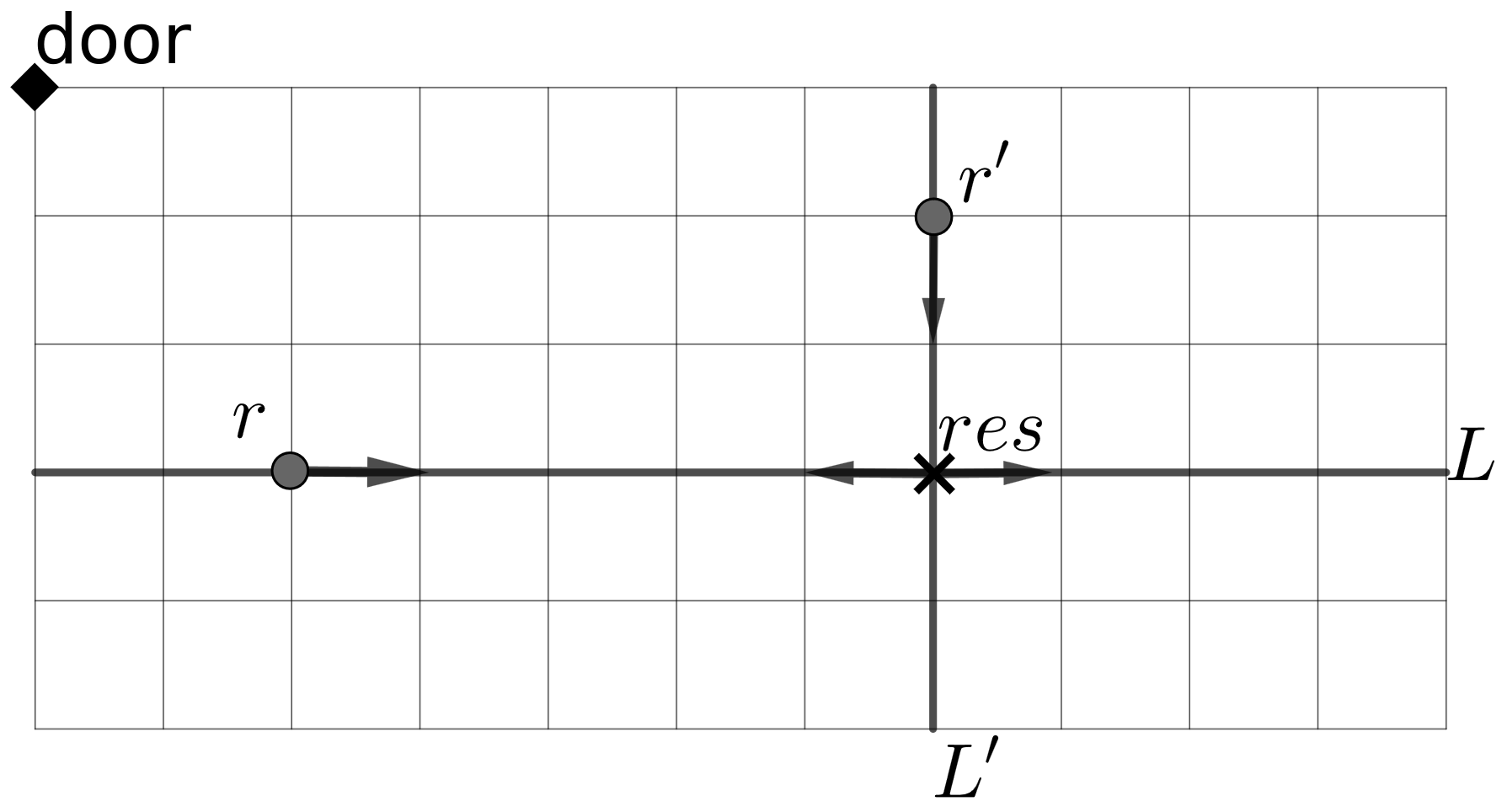}
    \caption{$r'$ stays on $L'$ but $res$ moves to either $L^{-1}$ or $L^1$ during round $t$. $r$ stays on $L$ with $res$.}
    \label{Fig:lemma1pic2case1}
\end{minipage}
\end{figure}

\textit{Case II:} Let res moves along $L'$ during the round $t$. In this case irrespective of the location of $r'$ at the beginning of round $t$ (i.e., $L^{-1}, L', L^1$), it moves onto $L'$ along with $res$ in the new configuration after completion of the round. Let $L''$ be the line perpendicular to $L'$ on which $res$ moves after completion of round $t$. Also note that during round $t$, $r$ stays on $L$ even if it moves. Now since, $L''$ is parallel to $L$ and one hop away from $L$, the perpendicular distance of $r$ to $L''$ becomes one after completion of the round $t$ (Fig.~\ref{Fig:Lemma1pic1case2} and Fig.~\ref{Fig:Lemma1pic2case2}). Thus the new configuration remains an \textsc{initGather Configuration} after completion of round $t$. 

\begin{figure}[h!]
\begin{minipage}[ht]{0.45\linewidth}
\centering
\includegraphics[width=6cm, height=4cm]{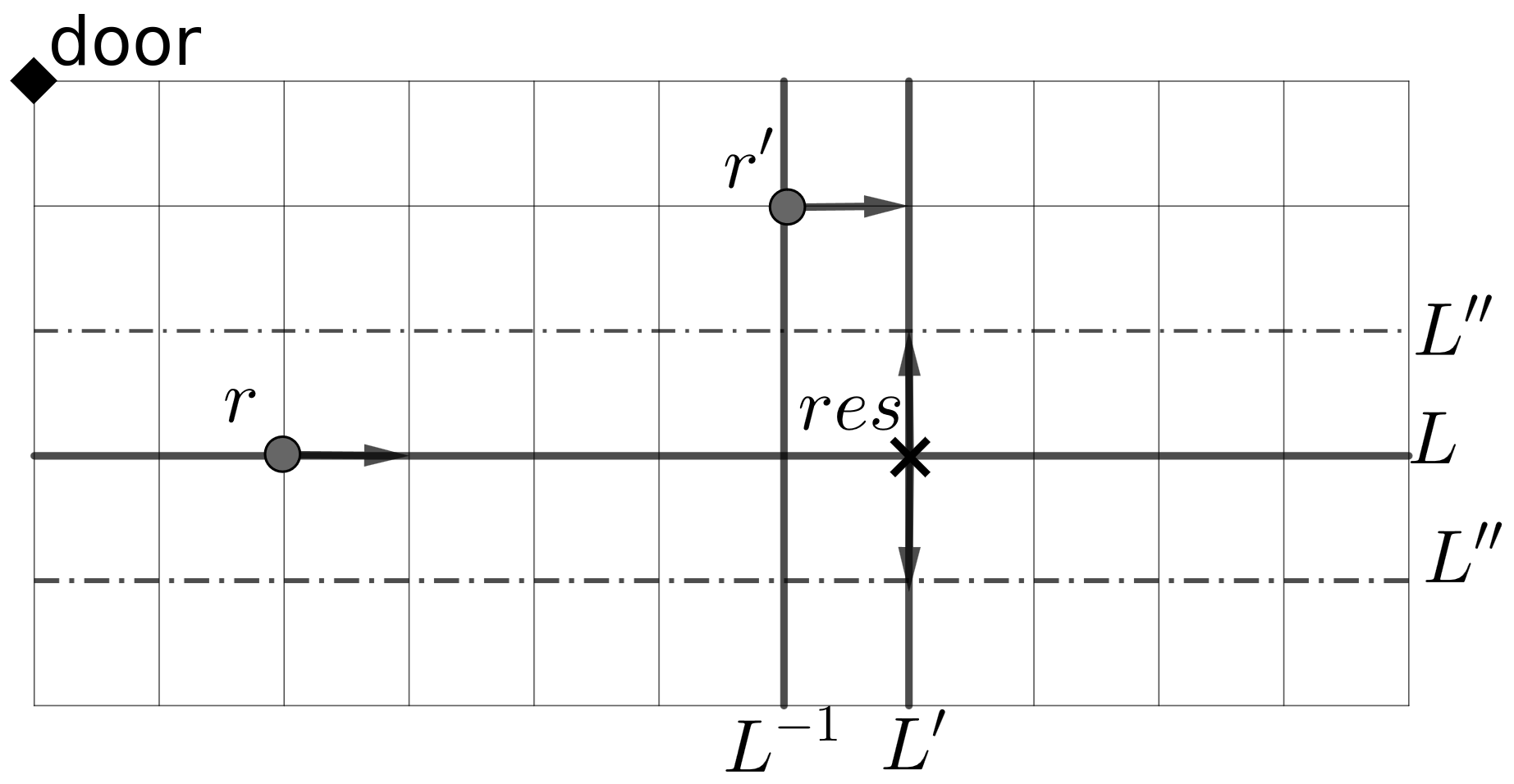}
     \caption{$res$ moves to $L''$ along $L'$, $r$ stays on $L$ and $r'$ moves to $L'$ from  $L^{-1}$.}
     \label{Fig:Lemma1pic1case2}
\end{minipage}
\hfill
\begin{minipage}[ht]{0.45\linewidth}
\centering
\includegraphics[height=4cm, width=6cm]{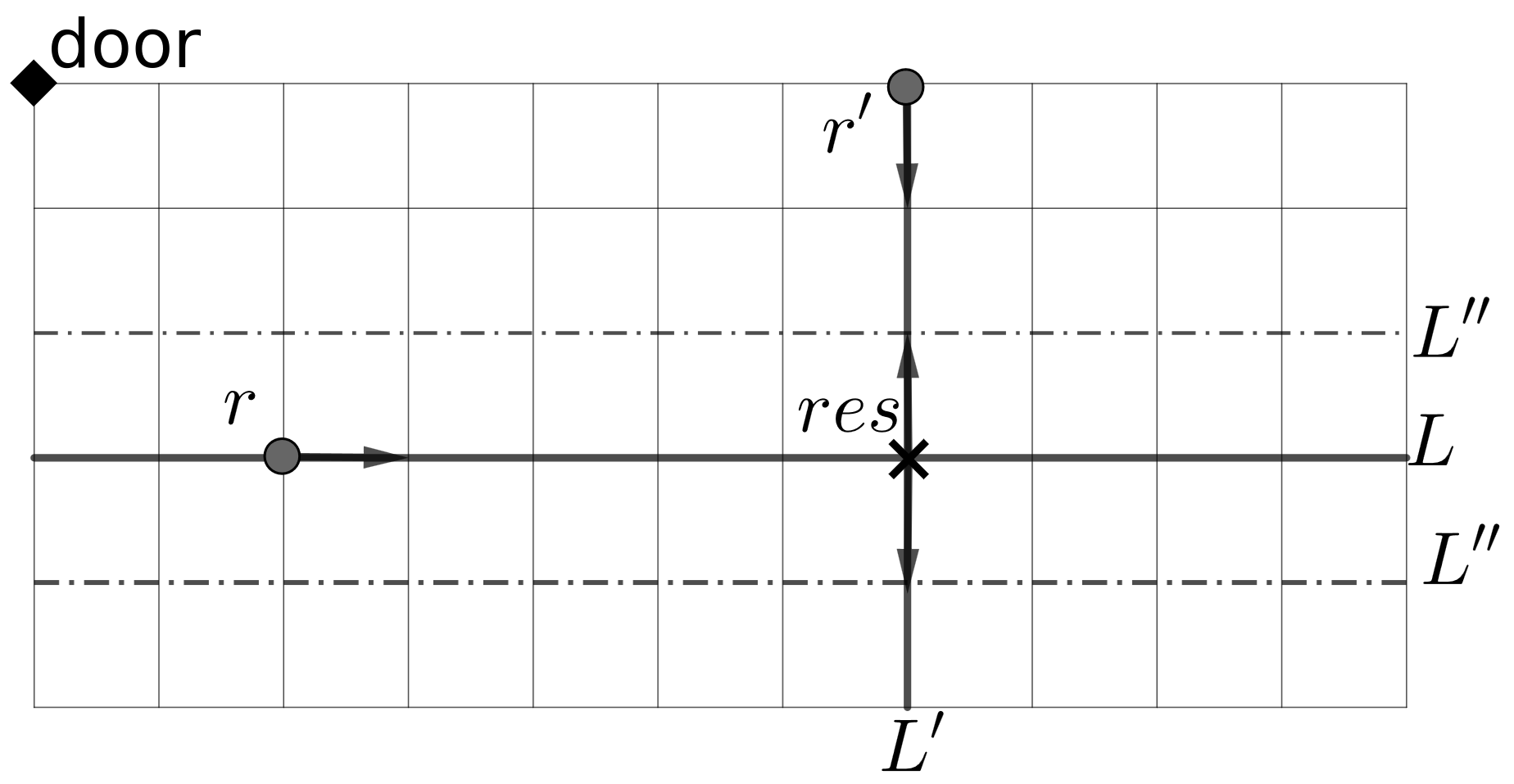}
    \caption{$res$ moves to $L''$ along $L'$, $r$ stays on $L$ and $r'$ stays on $L'$.}
    \label{Fig:Lemma1pic2case2}
\end{minipage}
\end{figure}

\textit{Case III:} Let $res$ does not move during round $t$. Then, irrespective of the location of $r'$ at the beginning of round $t$, it reaches $L'$ along with $res$ after completion of the round. Also, even if $r$ moves it stays on $L$ along with $res$ after completion of round $t$. So, after completion of round $t$, in the new configuration $r$ stays on $L$ along with $res$ and $r'$ stays on $L'$. Hence the new configuration is again an \textsc{InitGather Configuration}.

For all of the above cases, the configuration remains an \textsc{InitGather Configuration} and we have the lemma.
     \qed
 \end{proof}

 Let at the beginning of a particular round during the \textsc{Gather Phase} a robot $r$ is on the same line $L$ along with the resource, $res$. Let us define two lines, firstly, $L_1$ passing through $r$ and perpendicular to $L$, and secondly $L_2$, passing through the vertex of the other robot $r'$ and parallel to $L$. Note that the lines $L_1$ and $L_2$ divides the entire grid into one or more rectangles. The rectangle inside of which the resource $res$ is located is called the "Containing Rectangle" and it is denoted as $R_{Con}$ (Fig.~\ref{fig:RCon}). Observe that, at the beginning of the first round of \textsc{Gather Phase}, $L_1$ is $BD(r)$ and $L_2$ is $BD(r')$ and $R_{Con} = G$.  
 
 \begin{figure}[ht!]
    \centering
    \includegraphics[width=.5\textwidth]{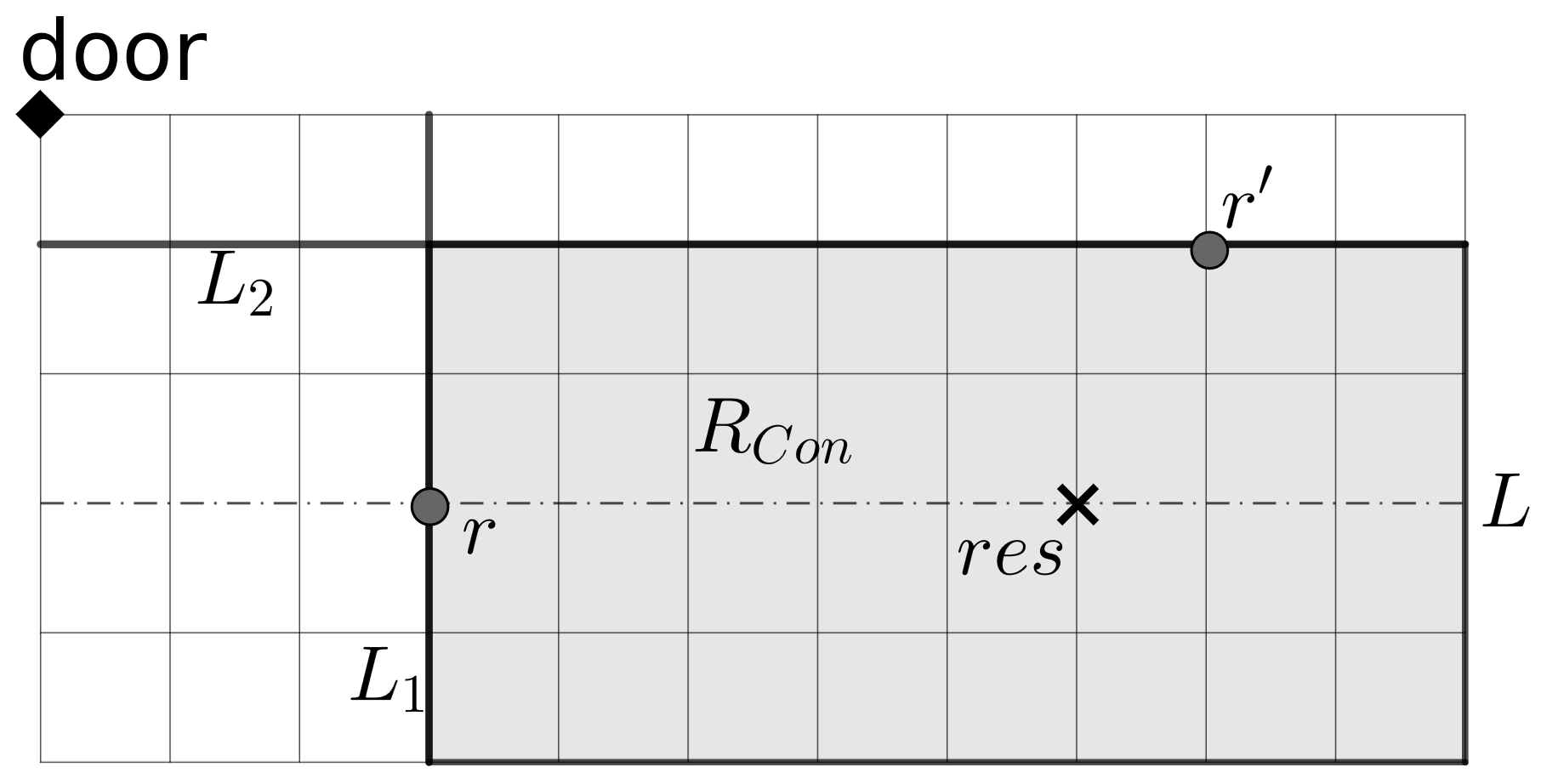}
    \caption{shaded region is $R_{Con}$}
    \label{fig:RCon}
\end{figure}

 \begin{lemma}
     \label{Lemma:resXl1l2}
     The resource $res$ never moves onto $L_1$ or $L_2$ during the \textsc{Gather Phase} without colliding with any robot.
 \end{lemma}
 \begin{proof}
 Initially at the beginning of \textsc{Gather Phase}, resource $res$ must not be on $BD(r)$ or $BD(r')$. Otherwise, $res$ must have crossed or moved onto both $PD(r)$ and $PD(r')$ during the \textsc{Boundary Phase}, which is not possible due to claim (3). So, at the beginning of \textsc{Gather Phase}, $res$ is not on any of $L_1$ or $L_2$.

If possible let during some round $t$ of \textsc{Gather Phase}, $res$ moves onto either $L_1$ or $L_2$ for the first time, without colliding with any robots. Let, $res$ is on the same line  $L$ with $r$ at the beginning of round $t$. Now, $r$ must be on the vertex $L_1 \cap L$ and $r'$ is on any of the three vertices, $L_2 \cap L^{-1}, L_2 \cap L'$ and $L_2 \cap L^1$ at the beginning of round $t$ (Fig.~\ref{fig:l3}). Now we have two cases,

\textit{Case I:} Let $res$ moves onto $L_1$ during round $t$. For this to happen, $res$ must be on the adjacent vertex of $L_1 \cap L$ (i.e, location of $r$) at the beginning of round $t$ and it must move along $L$ during round $t$.So, $res$ reaches $L_1 \cap L$ after completion of the round. Now since $r$ is adjacent to $res$ on $L$ at the beginning of round $t$, it does not move during round $t$ (Algorithm~\ref{algo:GathrPhase}) and stays on $L_1 \cap L$ after completion of the round. So after completion of round $t$, $res$ collides with $r$ contrary to our assumption.

\textit{Case II:} Let $res$ moves onto $L_2$ during round $t$. For this to happen, $res$ must move along $L'$ and reaches $L' \cap L_2$ after completion of the round. Now, irrespective of the position of $r'$ at the beginning of round $t$, it reaches $L_2 \cap L'$ (Algorithm~\ref{algo:GathrPhase}) after completion of the round (Fig.~\ref{fig:l3}). Since both $res$ and $r'$ reaches $L_2 \cap L'$ after completion of round $t$, they collides contradicting our assumption. 

Since in both cases we reach contradiction our assumption must be incorrect. So, $res$ never moves onto $L_1$ or $L_2$ without colliding with a robot during \textsc{Gather Phase}. 
     \qed
 \end{proof}
 
  \begin{figure}[ht!]
    \centering
    \includegraphics[width=.6\textwidth]{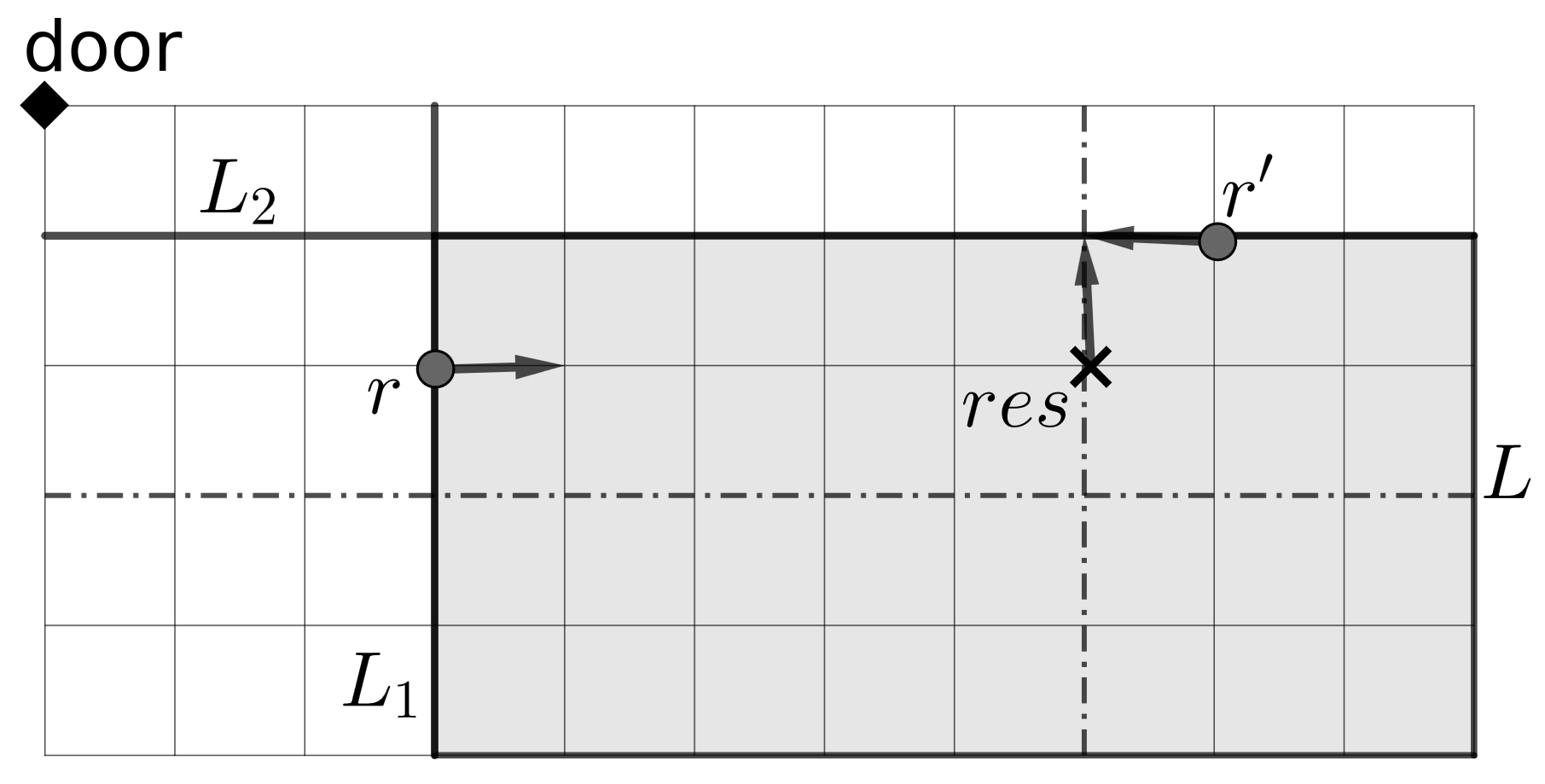}
    \caption{$res$ collides with $r'$ if moves onto $L_2$}
    \label{fig:l3}
\end{figure}
 \begin{corollary}
 The resource $res$ never moves outside $R_{Con}$ during the \textsc{Gather Phase} without colliding with a robot.
 \end{corollary}
 \begin{proof}
    If $res$ moves out of $R_{Con}$ then it must cross either $L_1$ or $L_2$ without moving onto them.
    
     Now the resource, $res$ never crosses $L_1$ without colliding with $r$. Also, for the same reason, $res$ never crosses $L_2$ while $r'$ is also on the same line along with $res$.  So let us assume $r'$ is not on the same line along with $res$  at the beginning of some round $t$ during which $res$ crosses $L_2$. So, during round $t$ $r'$ must move along $L_2$ . So, the line $L_2$ does not shift after the completion of round $t$. This implies $res$ must move onto $L_2$ to cross it which is not possible due to Lemma~\ref{Lemma:resXl1l2}. Hence $res$ never moves out of $R_{Con}$.   \qed
 \end{proof}

 Let at the beginning of the first round of \textsc{Gather Phase}, $R_{Con}$ be a $m_1 \times n_1$ grid. Let the height and width of $R_{Con}$ be $m_1$ and $n_1$ respectively where, both $m_1 >2$ and $n_1 >2$. We will prove that within $T_f+1$ rounds either both $m_1$ and $n_1$ decrease or one of $m_1$ and $n_1$ decreases and the other one stays the same.
 \begin{lemma}
     \label{Lemma:decreaseArea}
     During the \textsc{Gather Phase}, if both height and width of the $R_{Con}$ be more than two and none of the robot terminates then, within $T_f+1$ rounds either both height and width of $R_{Con}$ decreases or one of height or width decreases and the other remains same.  
 \end{lemma}
 \begin{proof}
     Let at the beginning of some round $t$ during the \textsc{Gather Phase}, $r$ be a robot on the line $L$ along with the resource, $res$. The lines $L_1$ (line passing through $r$ and perpendicular to $L$)and $L_2$ (line passing through the other robot, $r'$ and parallel to $L$) and the boundaries that do not contain the door vertex forms a rectangle $R_{Con}$.  We have proved that $res$ always remains contained within $R_{Con} \setminus \{L_1 \cup \L_2\}$ and never moves out of it if none of the robots are terminated. Let the dimension of $R_{Con}$ at the beginning of round $t$ be $m_1 \times n_1$ where both $m_1$ and $n_1$ are greater than two. Thus  even if $res$ is at a corner at the beginning of round $t$, both $r$ and $r'$ are not adjacent to $res$. Thus during round $t$, no robot moves to the location of $res$. 
     
     \textit{Case I:}
     Let at the beginning of round $t$, $r$ is not adjacent to $res$. Also, without loss of generality let the length of the side of $R_{Con}$, which is parallel to $L$ at the beginning of round $t$ is the width of $R_{Con}$. Now according to the algorithm, $r$ moves along $L$ towards $res$ i.e towards the direction of the interior of $R_{Con}$. Hence $L_1$ shifts towards the interior of $R_{Con}$. So the width of $R_{Con}$ decreases during round $t$. Now, if $r'$ is not on the line $L'$ (line passing through $res$ and perpendicular to $L$) or, adjacent to $res$ on the line $L'$ at the beginning of round $t$ then, $r'$ moves along $L_2$ (Fig.~\ref{Fig:Lemma4c1pic1}) or does not move at all. In both of these cases, the height of $R_{Con}$ remains the same after the completion of the round. on the other hand if at the beginning of round $t$, $r'$ is on $L'$ along with $res$ and not adjacent to $r'$  then, $r'$ moves along $L'$ towards the direction of $res$ (Fig.~\ref{Fig:Lemma4c1pic2}). Note that in this case $L_2$ also shifts towards the interior of $R_{Con}$ and decreases the height of $R_{Con}$ after completion of round $t$.
     
     \begin{figure}[h!]
\begin{minipage}[ht]{0.45\linewidth}
\centering
\includegraphics[width=5.5cm, height=3.5cm]{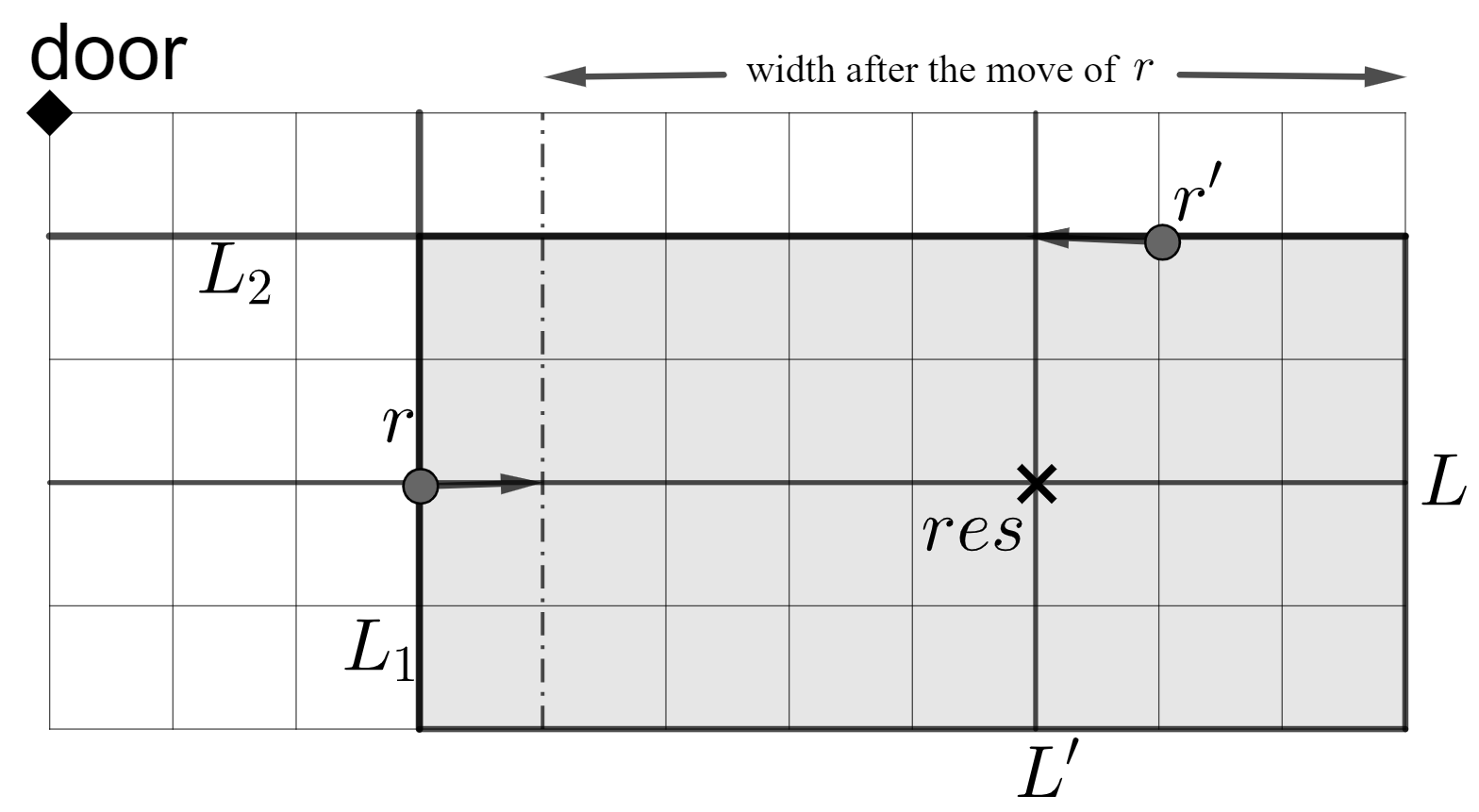}
     \caption{Only width of $R_{Con}$ decreases and height remains same.}
     \label{Fig:Lemma4c1pic1}
\end{minipage}
\hfill
\begin{minipage}[ht]{0.45\linewidth}
\centering
\includegraphics[height=3.5cm, width=6cm]{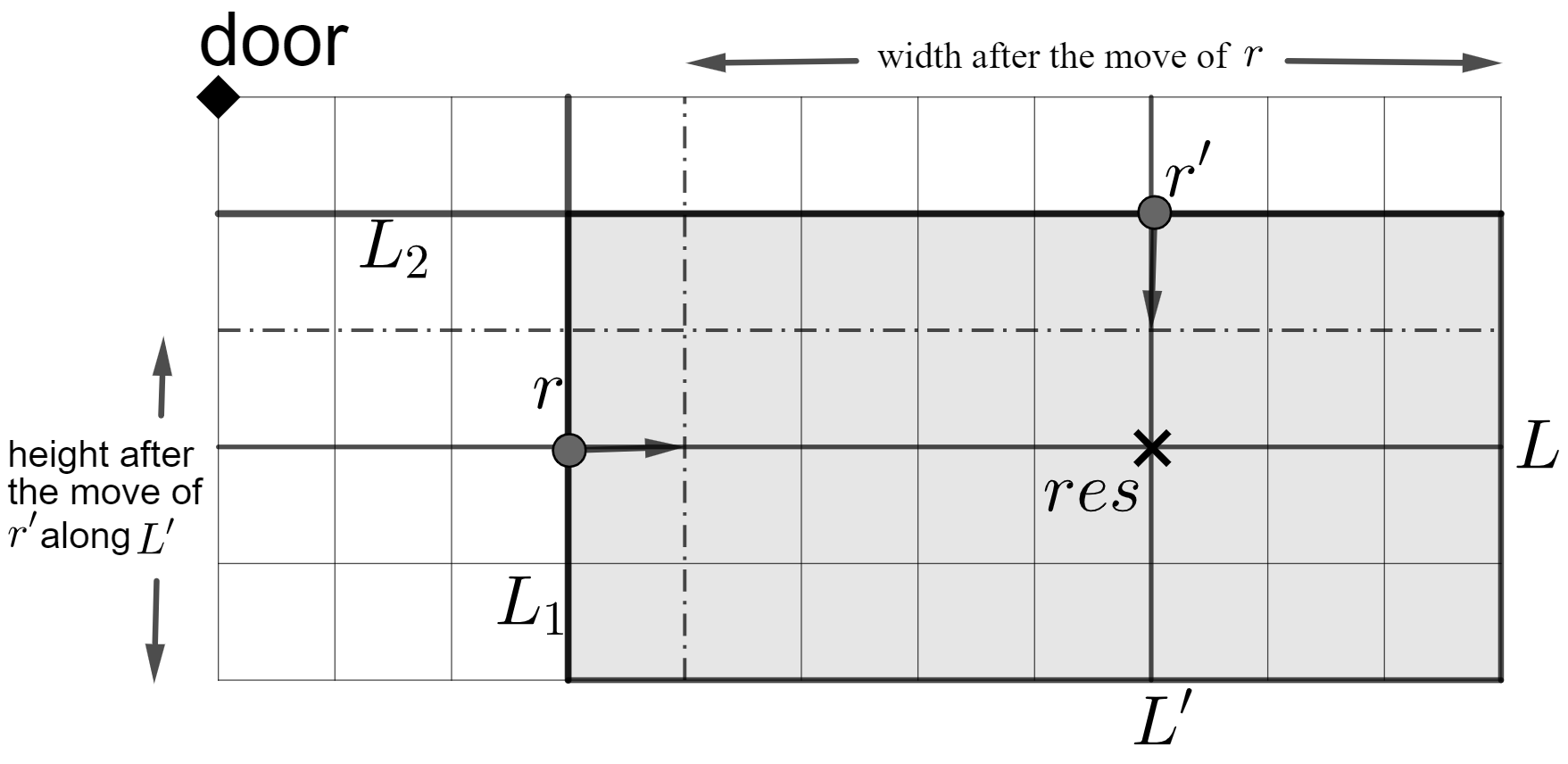}
    \caption{both height and width of $R_{Con}$ decreases.}
    \label{Fig:Lemma4c1pic2}
\end{minipage}
\end{figure}
     
     So we have shown that if $r$, a robot on a line $L$ with $res$ is not adjacent to $res$ then either both height and width decrease or only width decreases in one round.
     
     \textit{Case II:} Let,  $res$ is adjacent to $r$ on $L$ at the beginning of round $t$ then, $r$ will not move along $L$. It is assumed that $res$ will not stay at the same location for more than $T_f$ consecutive rounds. Note that since both height and width are more than two, $res$ gets an empty vertex to move. Now in the worst case during the round $t+T_f$, $res$ must have moved either along $L$ or along $L'$. Note that $res$ can not move towards $r$ along $L$ as it would end up colliding with $r$.
     
     \textit{Case II(a):} Let during the round $t+T_f$, $res$ moves along $L$ opposite to $r$ then, at the beginning of round $t+T_f+1$,  $r$ and $res$ are not adjacent along $L$. Hence during this round, either only the width of $R_{Con}$ decreases and height remains the same, or both height and width of $R_{Con}$  decrease by a similar argument as in \textit{case I}. 
     
     \textit{Case II(b):} Now let us consider the case where $r$ is adjacent to $res$ on $L$ at the beginning of round $t+T_f$ but $res$ moves perpendicular to $L$ i.e., along $L'$ during the round $t+T_f$.
    
     If at the beginning of round $t+T_f$, $r'$ was on $L'$ then, by the same argument as in \textit{Case I} and \textit{Case II(a)} we can conclude that in the worst case, after completion of round $t+2T_f+1$, the height of $R_{Con}$ must decrease while width either remains same or also decreases. 
    
    Let us now consider $r'$ is not on $L'$ at the beginning of round $t+T_f$. In this case, after completion of round $t+T_f$, $r'$ and $res$ must be on the line $L'$ and hence according to the same argument as above cases in the worst during round $t+2T_f+1$ either height and width of $R_{Con}$ both decrease or height decreases while the width remains same.
      \qed
 \end{proof}
 % \begin{corollary}
 % \label{cor:nodoorGather}
 % no robot moves to door vertex during \textsc{Gather Phase}.
 % \end{corollary}
 % \begin{proof}
 %     Let at round $t$, \textsc{Gather Phase} is started. Note that at the beginning of round $t$, no robot is on the door vertex as no robot moves to the door vertex during \textsc{Boundary Phase}. Now, after completion of round $t$, at least one robot must leave the boundary and moves inside $R_{Con}$. So from round $t+1$ onward, door vertex must remain outside of $R_{Con}$. This is because by lemma~\ref{Lemma:decreaseArea} height or width of $R_{Con}$ never increase to include the door vertex inside it. Hence the result.\qed
 % \end{proof}

 By Lemma~\ref{Lemma:InitGAther} we can conclude that, after one execution of the \textsc{Gather Phase}, if no robot is terminated then in the next round, \textsc{Gather Phase} will be executed again.  
Also, from the Lemma~\ref{Lemma:decreaseArea} it is evident that if the dimension of the $R_{Con}$ is $m_1 \times n_1$ where both $m_1 >2$ and $n_1> 2$ then, in the worst case in every $2T_f+1$ round, the height or the width of the configuration decreases and none of them ever increase. So within $O(T_f \times (m+n))$ there will be a round (say $t_0$) when either the height or the width of $R_{Con}$ becomes two. Without loss of generality let the dimension of $R_{Con}$ be $ m_1 \times 2$, at the beginning of round $t_0$, where $m_1 > 2$. Now we claim the following lemma.
\begin{lemma}
    \label{Lemma:areaFour}
    If the dimension of $R_{Con}$ is $ m_1 \times 2$ (resp. $2 \times n_1$) and no robot terminates then, within $(T_f+1)(m_1-1)$ (resp. $(T_f+1)(n_1-1)$) rounds dimension of $R_{Con}$ becomes $2 \times 2$.
\end{lemma}
\begin{proof}
 Without loss of generality let the dimension of $R_{Con}$ be $m_1 \times 2$ at the beginning of some round (say $t_0$) where, $m_1 > 2$. This implies exactly one of the height or width of $R_{Con}$ is two. Without loss of generality let the width is two and height be $m_1 >2$. Thus $R_{Con}$ consists of exactly two lines perpendicular to the width. One of these two lines (say, $L$) is a boundary of $G$ which does not contains the door vertex and the other one is the line parallel and adjacent to it (Say $L_2$). Since no robots are terminated, the configuration at the beginning of round $t_0$ is an \textsc{InitGather Configuration} (Lemma~\ref{Lemma:InitGAther}). So,  each of these two lines $L$ and $L_2$ contains exactly one robot. Let without loss of generality $r$ be on the line $L$ and $r'$ is on $L_2$. Now by Lemma~\ref{Lemma:resXl1l2}, $res$ must be on $L$ and below $L_1$. Also, the distance of $r'$ to the line perpendicular to $L$ and passing through $res$ (say $L'$) is at most one as the configuration is an \textsc{InitGather Configuration} at the beginning of round $t_0$. So if during round $t_0$, $res$ moves perpendicular to $L$ it must collides with $r'$ and $r'$ terminates contrary to the assumption. So let us consider $res$ either move along $L$ or does not move at all during round  $t_0$ (Fig.~\ref{fig:l5}). Now, in the worst case within $T_f+1$ rounds the height of $R_{Con}$ must decrease by one unit. Also, since $m_1 >2$,  at the beginning of round $t_0$, if $res$ is at a corner, $r$ is not adjacent to $res$ and hence $r'$ can only move parallel to $L_2$. So unless $m_1 =2$, width of $R_{Con}$ remains two. Thus in the worst case, the height of $R_{Con}$ becomes two while the width still remains two in $(T_f+1)(m_1-1)$ rounds. Hence the lemma. \qed 
\end{proof}

\begin{figure}[ht!]
    \centering
    \includegraphics[width=.6\textwidth]{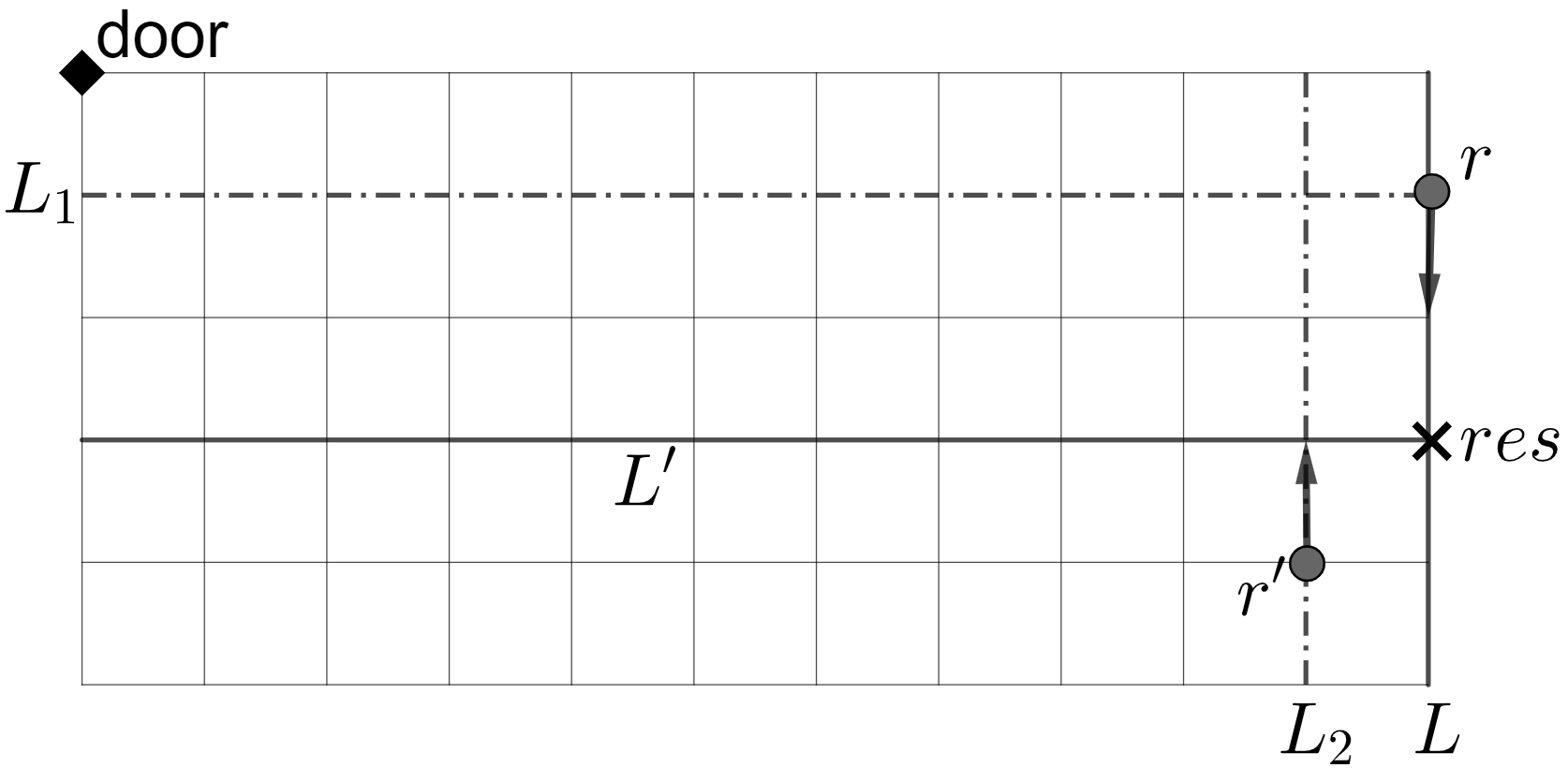}
    \caption{Dimension of decreases from $m_1 \times 2$ to $(m_1-1)\times 2$ where $m_1 >2$. }
    \label{fig:l5}
\end{figure}

Now we have proved that if no robot terminates then, within $O(T_f\times (m+n))$ rounds there is a round $t_1$ such that at the beginning of it, $R_{Con}$ is a $2\times 2$ rectangle on the bottom right corner of the grid $G$ (Fig.~\ref{fig:final}). At the beginning of round $t_1$, $res$ must be at the corner of the Grid diagonally opposite to the door vertex (by Lemma~\ref{Lemma:resXl1l2}). Here two robots $r$ and $r'$ must be on two different adjacent vertices of $res$ as the configuration is an \textsc{InitGather Configuration} (Lemma~\ref{Lemma:InitGAther}). Hence by the algorithm of \textsc{Gather Phase} $r$ and $r'$ both move to the vertex of $res$ during the round $t_1$, while $res$ has no other edges to move out as it is on the corner. So, both robot reaches the location of $res$ and terminates.
From this discussion, we can conclude the following Theorem.

\begin{figure}[ht!]
    \centering
    \includegraphics[width=.6\textwidth]{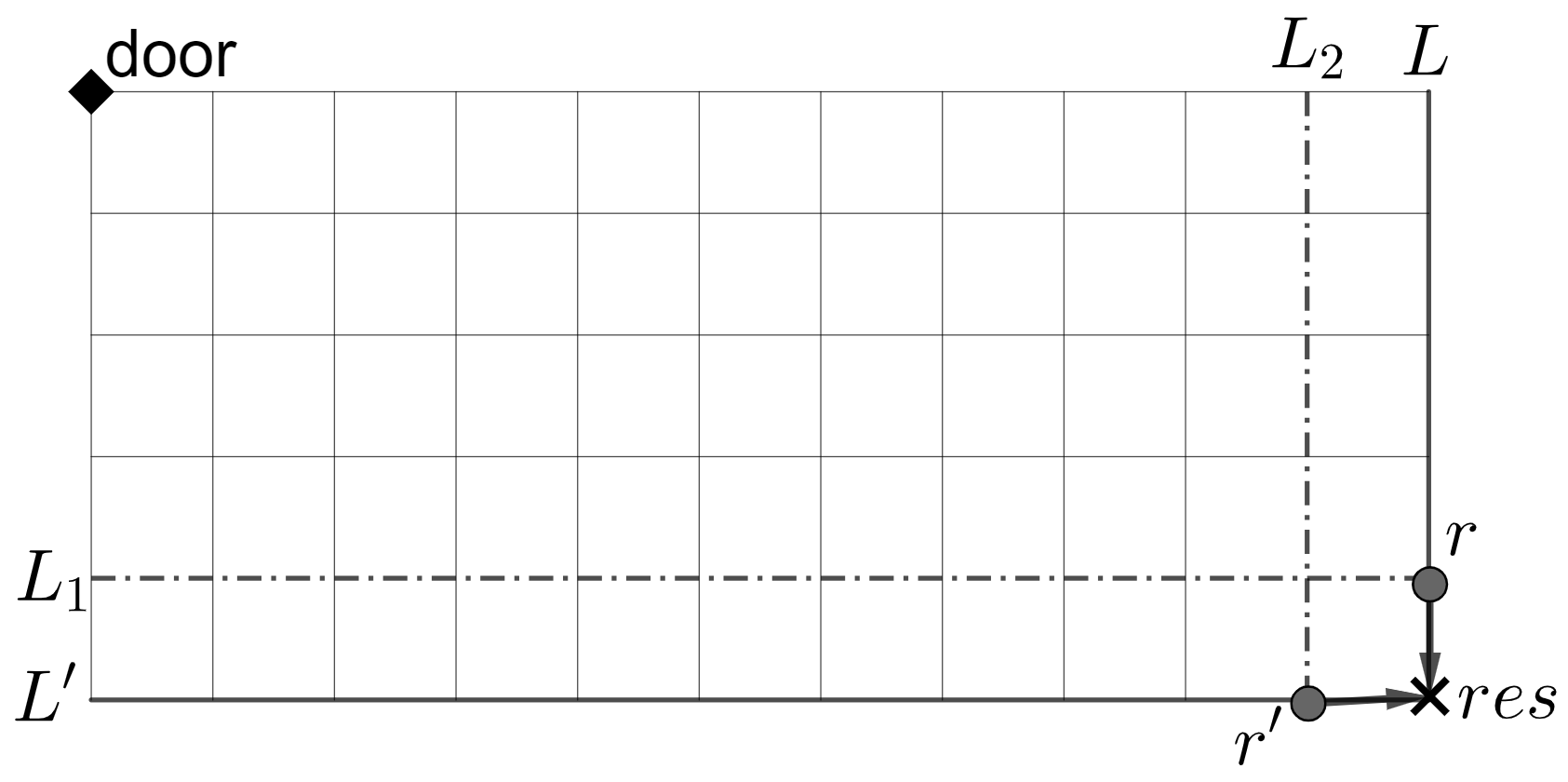}
    \caption{$R_{Con}$ has dimension $2 \times 2$. }
    \label{fig:final}
\end{figure}
\begin{theorem}
    For a grid of dimension $m \times n$, the \textsc{Gather Phase} terminates within $O(T_f\times (m+n))$ rounds.
\end{theorem}
Now since the termination of \textsc{Gather Phase} implies termination of the whole algorithm we can conclude with the following theorem.
\begin{theorem}
    Algorithm \textsc{Dynamic Rendezvous} terminates within $O(T_f \times (m+n))$ rounds.
\end{theorem}
\section{Conclusion}
\label{sec:5}
Gathering is a classical problem in the field of swarm robotics. Rendezvous is a special case of gathering where two robots gather at a single point in the environment. All the previous works on gathering considered the meeting point to be not known by the robots but here we have considered the robots to know the meeting point but the meeting point can move in the environment until a robot reaches it. To the best of our knowledge, it is the first work that considers a dynamic meeting point. In this work, we have shown that it is impossible for two robots to gather at a known dynamic meeting point on a finite grid if the scheduler is semi-synchronous. Then considering a fully synchronous scheduler we have provided a distributed algorithm \textsc{Dynamic Rendezvous} which gathers the two robots on the known dynamic meeting point called the resource, within $O(T_f \times (m+n))$ rounds where $m \times n$ is the dimension of the grid and $T_f$ is the upper bound of the number of consecutive rounds the resource can stay at a single vertex alone. We have also provided a lower bound of time i.e., $\Omega(m+n)$ to solve this problem considering a $m \times n$ grid. So, if $T_f \le k$ for some constant $k$ then our algorithm is time optimal.

For future courses of research, one can think of solving this problem on other different networks such as tree, ring, etc. In ring networks, solving this problem with limited visibility can be really interesting. Also, One can think of finding out the minimum number of robots needed to gather at a known dynamic meeting point for different schedulers in different networks. 
%

% ---- Bibliography ----
%
% BibTeX users should specify bibliography style 'splncs04'.
% References will then be sorted and formatted in the correct style.
%
 \bibliographystyle{splncs04}
 \bibliography{samplepaper}

\begin{thebibliography}{10}
\providecommand{\url}[1]{\texttt{#1}}
\providecommand{\urlprefix}{URL }
\providecommand{\doi}[1]{https://doi.org/#1}

\bibitem{BAKS19}
Bose, K., Adhikary, R., Kundu, M.K., Sau, B.: Arbitrary pattern formation on
  infinite grid by asynchronous oblivious robots. In: Das, G.K., Mandal, P.S.,
  Mukhopadhyaya, K., Nakano, S. (eds.) {WALCOM:} Algorithms and Computation -
  13th International Conference, {WALCOM} 2019, Guwahati, India, February 27 -
  March 2, 2019, Proceedings. Lecture Notes in Computer Science, vol. 11355,
  pp. 354--366. Springer (2019). \doi{10.1007/978-3-030-10564-8\_28},
  \url{https://doi.org/10.1007/978-3-030-10564-8\_28}

\bibitem{CFPS12}
Cieliebak, M., Flocchini, P., Prencipe, G., Santoro, N.: Distributed computing
  by mobile robots: Gathering. {SIAM} J. Comput.  \textbf{41}(4),  829--879
  (2012). \doi{10.1137/100796534}, \url{https://doi.org/10.1137/100796534}

\bibitem{CP05}
Cohen, R., Peleg, D.: Convergence properties of the gravitational algorithm in
  asynchronous robot systems. {SIAM} J. Comput.  \textbf{34}(6),  1516--1528
  (2005). \doi{10.1137/S0097539704446475},
  \url{https://doi.org/10.1137/S0097539704446475}

\bibitem{DSKN16}
D'Angelo, G., Stefano, G.D., Klasing, R., Navarra, A.: Gathering of robots on
  anonymous grids and trees without multiplicity detection. Theor. Comput. Sci.
   \textbf{610},  158--168 (2016). \doi{10.1016/j.tcs.2014.06.045}

\bibitem{DSN14}
D'Angelo, G., Stefano, G.D., Navarra, A.: Gathering six oblivious robots on
  anonymous symmetric rings. J. Discrete Algorithms  \textbf{26},  16--27
  (2014). \doi{10.1016/j.jda.2013.09.006}

\bibitem{DBS21}
Das, A., Bose, K., Sau, B.: Memory optimal dispersion by anonymous mobile
  robots. In: Mudgal, A., Subramanian, C.R. (eds.) Algorithms and Discrete
  Applied Mathematics - 7th International Conference, {CALDAM} 2021, Rupnagar,
  India, February 11-13, 2021, Proceedings. Lecture Notes in Computer Science,
  vol. 12601, pp. 426--439. Springer (2021).
  \doi{10.1007/978-3-030-67899-9\_34},
  \url{https://doi.org/10.1007/978-3-030-67899-9\_34}

\bibitem{DFPSY16}
Das, S., Flocchini, P., Prencipe, G., Santoro, N., Yamashita, M.: Autonomous
  mobile robots with lights. Theor. Comput. Sci.  \textbf{609},  171--184
  (2016). \doi{10.1016/j.tcs.2015.09.018},
  \url{https://doi.org/10.1016/j.tcs.2015.09.018}

\bibitem{FPSW05}
Flocchini, P., Prencipe, G., Santoro, N., Widmayer, P.: Gathering of
  asynchronous robots with limited visibility. Theoretical Computer Science
  \textbf{337}(1),  147--168 (2005).
  \doi{https://doi.org/10.1016/j.tcs.2005.01.001}

\bibitem{FSVY16}
Flocchini, P., Santoro, N., Viglietta, G., Yamashita, M.: Rendezvous with
  constant memory. Theor. Comput. Sci.  \textbf{621},  57--72 (2016).
  \doi{10.1016/j.tcs.2016.01.025},
  \url{https://doi.org/10.1016/j.tcs.2016.01.025}

\bibitem{GSGS22}
Goswami, P., Sharma, A., Ghosh, S., Sau, B.: Time optimal gathering of myopic
  robots on an infinite triangular grid. In: Devismes, S., Petit, F., Altisen,
  K., Luna, G.A.D., Anta, A.F. (eds.) Stabilization, Safety, and Security of
  Distributed Systems - 24th International Symposium, {SSS} 2022,
  Clermont-Ferrand, France, November 15-17, 2022, Proceedings. Lecture Notes in
  Computer Science, vol. 13751, pp. 270--284. Springer (2022).
  \doi{10.1007/978-3-031-21017-4\_18},
  \url{https://doi.org/10.1007/978-3-031-21017-4\_18}

\bibitem{HDT18}
Heriban, A., D{\'{e}}fago, X., Tixeuil, S.: Optimally gathering two robots. In:
  Bellavista, P., Garg, V.K. (eds.) Proceedings of the 19th International
  Conference on Distributed Computing and Networking, {ICDCN} 2018, Varanasi,
  India, January 4-7, 2018. pp. 3:1--3:10. {ACM} (2018).
  \doi{10.1145/3154273.3154323}, \url{https://doi.org/10.1145/3154273.3154323}

\bibitem{ISKIDWY12}
Izumi, T., Souissi, S., Katayama, Y., Inuzuka, N., D{\'{e}}fago, X., Wada, K.,
  Yamashita, M.: The gathering problem for two oblivious robots with unreliable
  compasses. {SIAM} J. Comput.  \textbf{41}(1),  26--46 (2012).
  \doi{10.1137/100797916}, \url{https://doi.org/10.1137/100797916}

\bibitem{KKN10}
Klasing, R., Kosowski, A., Navarra, A.: Taking advantage of symmetries:
  Gathering of many asynchronous oblivious robots on a ring. Theor. Comput.
  Sci.  \textbf{411}(34-36),  3235--3246 (2010).
  \doi{10.1016/j.tcs.2010.05.020}

\bibitem{KMP08}
Klasing, R., Markou, E., Pelc, A.: Gathering asynchronous oblivious mobile
  robots in a ring. Theor. Comput. Sci.  \textbf{390}(1),  27--39 (2008).
  \doi{10.1016/j.tcs.2007.09.032}

\bibitem{DUVY20}
Luna, G.A.D., Uehara, R., Viglietta, G., Yamauchi, Y.: Gathering on a circle
  with limited visibility by anonymous oblivious robots. In: Attiya, H. (ed.)
  34th International Symposium on Distributed Computing, {DISC} 2020, October
  12-16, 2020, Virtual Conference. LIPIcs, vol.~179, pp. 12:1--12:17. Schloss
  Dagstuhl - Leibniz-Zentrum f{\"{u}}r Informatik (2020).
  \doi{10.4230/LIPIcs.DISC.2020.12}

\bibitem{OT18}
Ooshita, F., Tixeuil, S.: Ring exploration with myopic luminous robots. CoRR
  \textbf{abs/1805.03965} (2018), \url{http://arxiv.org/abs/1805.03965}

\bibitem{PS21}
Poudel, P., Sharma, G.: Time-optimal gathering under limited visibility with
  one-axis agreement. Inf.  \textbf{12}(11), ~448 (2021).
  \doi{10.3390/info12110448}

\bibitem{P07}
Prencipe, G.: Impossibility of gathering by a set of autonomous mobile robots.
  Theor. Comput. Sci.  \textbf{384}(2-3),  222--231 (2007).
  \doi{10.1016/j.tcs.2007.04.023},
  \url{https://doi.org/10.1016/j.tcs.2007.04.023}

\bibitem{SIYM99}
Suzuki, I., Yamashita, M.: Distributed anonymous mobile robots: Formation of
  geometric patterns. SIAM J. Comput.  \textbf{28}(4),  1347–1363 (mar 1999).
  \doi{10.1137/S009753979628292X},
  \url{https://doi.org/10.1137/S009753979628292X}

\end{thebibliography}
\end{document}